\def\noeditingmarks{}  

\documentclass[10pt,twocolumn]{article}
\usepackage{usenix-style}

\usepackage{subcaption} 
\usepackage{booktabs}  
\usepackage{dcolumn}   
\usepackage{multirow}  
\usepackage{rotating}
\usepackage{xspace}
\usepackage{hyphenat}  
\usepackage{color}
\usepackage{enumerate}
\usepackage{balance}
\usepackage{fancyvrb}
\usepackage{comment}
\usepackage{changepage}
\usepackage{textcomp}
\usepackage{amsmath,lipsum}
\usepackage{amssymb}
\usepackage{amsfonts}
\usepackage{wasysym}
\usepackage{lastpage}
\usepackage{tabularx}
\usepackage{pifont}
\usepackage{microtype} 
\usepackage{mathtools} 
\usepackage{mathrsfs} 

\DeclarePairedDelimiter\floor{\lfloor}{\rfloor}
\usepackage{url}

\usepackage{ulem}
\usepackage{array}    

\usepackage{dblfloatfix}

\usepackage[noend]{algpseudocode}
\algrenewcomment[1]{\hfill// #1}
\algtext*{Function}
\algrenewcommand{\algorithmicrequire}{\textbf{Input:}}
\algrenewcommand{\algorithmicensure}{\textbf{Output:}}
\algnewcommand{\LineComment}[1]{\State // #1}
\algrenewcommand\textproc{}

\usepackage[noeka]{mathrmletter}
\usepackage{arydshln}
\usepackage[english]{babel}
\usepackage{blindtext}

\newif\ifminted
\newif\ifintrouble
\newif\ifcuttext

\ifintrouble
\cuttexttrue
\else
\cuttextfalse
\fi

\ifx\buildminted\undefined
\else
  \mintedtrue
\fi
\ifminted
  \usepackage{minted}
\else
\usepackage{colortbl} 
\fi

\newcommand{\textred}[1]{\begingroup \color{red} #1\endgroup}
\ifx\noeditingmarks\undefined
   \newcommand{\pgwrapper}[2]{\textred{#1: #2}}
   \newcommand{\pgwrapperb}[1]{\textbf{#1}}
\else
   \newcommand{\pgwrapperb}[1]{}
   \newcommand{\pgwrapper}[2]{}
\fi

\ifx\noeditingmarks\undefined
    \newcommand{\changebars}[2]{%
    [{\em \begingroup \color{magenta} #1 \endgroup}]
    [\begingroup \color{magenta} \sout{#2} \endgroup]}
    
\else
    \newcommand{\changebars}[2]{#1}
    
\fi

\newcommand{\sys}{{Otak}\xspace}
\newcommand{\Sys}{\sys}
\newcommand{\sysstee}{{Otak-\textsc{STEE}}\xspace}
\newcommand{\sysmtee}{{Otak-\textsc{MTEE}}\xspace}

\newcommand{\delphi}{\textsc{DELPHI}\xspace}
\newcommand{\spdz}{\textsc{SPDZ}\xspace}
\newcommand{\privado}{\textsc{PRIVADO}\xspace}

\makeatletter
\renewcommand*{\@fnsymbol}[1]{\ensuremath{\ifcase#1\or \star\or \dagger\or \ddagger\or
   \mathsection\or \mathparagraph\or \|\or **\or \dagger\dagger
   \or \ddagger\ddagger \else\@ctrerr\fi}}
\makeatother

\def\hn{\usefont{OT1}{phv}{mc}{n}\selectfont}

\newcommand{\mpfont}{\hn\scriptsize}
\newcommand{\MPworker}[2]{{\color{#1}\vrule\vrule}{\marginpar{\color{#1}\mpfont #2}}}
\ifx\noeditingmarks\undefined
    \newcommand{\MP}[1]{\MPworker{red}{#1}}
    \newcommand{\MPtg}[1]{\MPworker{red}{#1}}
    \newcommand{\MPva}[1]{\MPworker{blue}{#1}}
    \newcommand{\MPmn}[1]{\MPworker{green}{#1}}
    \newcommand{\MPag}[1]{\MPworker{brown}{#1}}
\else
   \newcommand{\MP}[1]{}
    \newcommand{\MPtg}[1]{}
     \newcommand{\MPmn}[1]{}
    \newcommand{\MPva}[1]{}
    \newcommand{\MPag}[1]{}
\fi
\newcommand\rmv[1]{}

\newcommand{\techReportOnly}[1]{}

\newcommand{\hssgen}{{\small\textsf{HSS.Gen}}\xspace}
\newcommand{\hssenc}{{\small\textsf{HSS.Enc}}\xspace}
\newcommand{\hsseval}{{\small\textsf{HSS.Eval}}\xspace}

\newcommand{\lprenc}{{\small\textsf{LPR.Enc}}\xspace}

\newcommand{\fssgen}{{\small\textsf{FSS.Gen}}\xspace}

\newcommand{\fsseval}{{\small\textsf{FSS.Eval}}\xspace}

\newcommand{\hssbkslprname}{\textsc{BKS-LPR}\xspace}
\newcommand{\bkslprgen}{{\small\textsf{BKS-LPR.Gen}}\xspace}
\newcommand{\bkslprenc}{{\small\textsf{BKS-LPR.Enc}}\xspace}
\newcommand{\bkslpreval}{{\small\textsf{BKS-LPR.Eval}}\xspace}

\usepackage{tkz-euclide}
\usepackage{tikz}

\usepackage{mathtools}
\DeclarePairedDelimiterX{\norm}[1]{\lVert}{\rVert}{#1}

\newcommand{\cpu}{\textsc{cpu}\xspace}
\newcommand{\cpus}{\textsc{cpu}s\xspace}

\newcommand{\myvec}[1]{\ensuremath{\mathit{\mathbf{#1}}}}
\newcommand{\mymatrix}[1]{\ensuremath{\mathit{\mathbf{#1}}}}

\newcommand{\func}{\mathcal{F}}
\newcommand{\view}{\mathsf{View}}
\newcommand{\adversary}{\mathcal{A}}
\newcommand{\simr}{\mathsf{Sim}}

\def\compactify{\itemsep0in \topsep2pt \parsep=0.00in \partopsep=0pt
\leftmargin2em}
\let\latexusecounter=\usecounter

\newenvironment{myenumerate}
  {\def\usecounter{\compactify\latexusecounter}
   \begin{enumerate}}
  {\end{enumerate}\let\usecounter=\latexusecounter}

\newenvironment{myenumerate3}
  {\def\usecounter{\itemsep0pt \topsep0pt \parsep=0ex
\partopsep=0pt \leftmargin1.5em\latexusecounter}
   \begin{enumerate}}
  {\end{enumerate}\let\usecounter=\latexusecounter}

\newenvironment{myenumerate4}
  {\def\usecounter{\itemsep=0ex \topsep1ex \parsep=1ex \partopsep=0pt
\leftmargin2em\latexusecounter}
   \begin{enumerate}}
  {\end{enumerate}\let\usecounter=\latexusecounter}

\newenvironment{myitemize}%
  {\begin{list}{\labelitemi}{\itemsep0in \topsep2pt \parsep0.00in
  \partopsep=0pt \leftmargin\parindent}}%
  {\end{list}}

\newenvironment{myitemize2}%
  {\begin{list}{\labelitemi}{\itemsep6pt \topsep6pt \parsep0.00in
  \partopsep=0pt \leftmargin\parindent}}%
  {\end{list}}

\newenvironment{myitemize3}%
  {\begin{list}{\labelitemi}{\itemsep3pt \topsep3pt \parsep0.00in
  \partopsep=0pt \leftmargin\parindent}}%
  {\end{list}}
  {\begin{list}{\labelitemi}{\itemsep0in \topsep2pt \parsep0.00in
  \partopsep=0pt \leftmargin\parindent}}%
  {\end{list}}

\newenvironment{myitemize5}%
  {\begin{list}{\labelitemi}{\itemsep3pt \topsep3pt \parsep0.00in
  \partopsep=0pt \leftmargin1em}}%
  {\end{list}}

\usepackage{amsthm}

\theoremstyle{definition}
\newtheorem{definition}{Definition}[section]

\theoremstyle{lemma}
\newtheorem{lemma}{Lemma}[section]

\theoremstyle{theorem}
\newtheorem{theorem}{Theorem}[section]

\theoremstyle{corollary}

\theoremstyle{conjecture}

\def\emparagraph#1{\vspace{0.25em}\noindent{\bf #1}}

\def\discretionaryslash{\discretionary{/}{}{/}}
{\catcode`\/\active
\gdef\URLprepare{\catcode`\/\active\let/\discretionaryslash
        \def~{\char`\~}}}%
\def\URL{\bgroup\URLprepare\realURL}%
\def\realURL#1{\tt #1\egroup}%

\begin{document}

\newcommand{\supsyml}[1]{\raisebox{4pt}{\footnotesize #1}}
\newcommand{\rstar}{\supsyml{$\ast$}\xspace}
\newcommand{\rdag}{\supsyml{$\ast\ast$}\xspace}

\author{
\fontsize{10.75}{12}\selectfont 
          Muqsit Nawaz$^{\ast}$\quad
          Aditya Gulati$^{\ast\dag}$\quad
          Kunlong Liu$^{\ast}$\quad
          Vishwajeet Agrawal$^{\ddag}$\quad
          Prabhanjan Ananth$^{\ast}$\quad
          Trinabh Gupta$^{\ast}$\quad \\[6pt]
\fontsize{9.5}{11}\selectfont 
$^{\ast}$UCSB
\quad\quad $^{\dag}$IIT Kanpur
\quad\quad $^{\ddag}$IIT Delhi
}

\date{}

\title{\textbf{Accelerating 2PC-based ML with Limited Trusted Hardware}}
\maketitle

\begin{abstract}
This paper describes the design, implementation, and evaluation of \sys, a
system that 
    allows two non-colluding cloud providers to run machine learning (ML)
        inference without knowing the inputs to inference.
Prior work for this problem mostly relies on advanced cryptography such as
two-party secure computation (2PC) protocols that provide rigorous guarantees but
suffer from high resource overhead. 
\Sys improves efficiency via a new 2PC
protocol that (i) tailors recent primitives such as function and
homomorphic secret sharing to ML inference, and (ii) uses trusted hardware in a limited capacity to bootstrap
the protocol. At the same time, \sys reduces trust assumptions on trusted
hardware by running a small code inside the hardware, restricting its use to a
preprocessing step, and distributing trust over heterogeneous trusted
hardware platforms from different vendors. An implementation and evaluation of
\sys demonstrates 
that its
\cpu and network overhead
converted to a dollar amount is 5.4--385$\times$ lower than state-of-the-art 
2PC-based works. Besides, \sys's
trusted computing base (code inside trusted hardware)
is only 1,300 lines of code, which is 14.6--29.2$\times$ lower than
the code-size in prior trusted hardware-based works.
\end{abstract}
\section{Introduction}
\label{s:intro}
\label{s:introduction}

How can a machine learning (ML) system running in the cloud perform inference 
without getting access to the inputs to inference 
(model parameters and the data points
whose class is being inferred)?

This question is motivated by a fundamental tension 
    between ease-of-use and confidentiality of user data.
On the one hand, cloud providers expose
easy-to-use ML APIs~\cite{googlemlpredictionapi, amazonml, azureml}.
A user can call them with model parameters and input data points, and receive 
inference results while treating ML as a black-box.
Furthermore, the user does not have to provision and manage ML systems locally. 
On the other hand, ML APIs require inputs in plaintext.  
Thus, a user's sensitive or proprietary model parameters and data points can be
    accessed 
    by rogue system
    administrators at the cloud provider~\cite{roguesteffan, roguegoogle, roguetechspot}, hackers who can
    get into the cloud provider's infrastructure~\cite{cloudhopper, cyber20}, and
    government agencies~\cite{propublica, googletransparency,
    amazontransparency, microsofttransparency}.

Given the wide array of ML applications, 
the need to balance the benefits and
    confidentiality-risks of cloud-hosted ML services has received significant attention
    (\S\ref{s:relwork}).
A long line of work
relies
    on cryptographic techniques~\cite{riazi2018chameleon, wagh2019securenn,
kumar2020cryptflow, wagh2020falcon, mohassel2018aby, mohassel2017secureml,
    riazi2019xonn,ball2019garbled,liu2017oblivious,
    juvekar2018gazelle,rouhani2018deepsecure,mishra2020delphi, chen2019secure,
    gilad2016cryptonets, xie2014crypto, hesamifard2017cryptodl,
    chabanne2017privacy,boemer2018ngraph, badawi2018alexnet,brutzkus2019low,
    chou2018faster,bourse2018fast,lou2019glyph, jiang2018secure}.
These works 
provide rigorous 
    guarantees but incur high resource overhead (\cpu consumption, network
    transfers, etc.). 
For example, for a single inference over the ResNet-32
model~\cite{he2016deep}, 
    state-of-the-art systems that run over two non-colluding cloud
providers~\cite{mohassel2017secureml,keller2018overdrive}   
    make over 6~GB of expensive, wide-area network transfers
(\S\ref{s:eval:overheadinference}).

In a quest to avoid expensive cryptography, 
researchers have resorted to using trusted execution environments (TEEs). 
A TEE
consists of a secure container 
that 
can execute a program 
such that an external entity peeking inside the container
learns only the input-output behavior of the computation. Secure systems developed using TEEs not only are less complex (and hence, easier to build) but also offer
significant efficiency benefits 
over their counterparts built only using cryptography. Indeed,
there are many highly efficient systems developed over the years~\cite{ohrimenko2016oblivious, hynes2018efficient,
  tople2018privado, hunt2018ryoan, hunt2018chiron,
tramer2018slalom, narra2019privacy, fischer2020computation}. For instance, the
\textsc{PRIVADO} system for ML inference 
    incurs less than 20\% overhead relative to
a non-private system~\cite{tople2018privado}.

However, TEEs are a not a panacea for building secure systems. Over the years, researchers have discovered various cryptanalytic attacks. First, the code inside a TEE can leak sensitive data through 
bugs and digital side-channels~\cite{moghimi2017cachezoom,gotzfried2017cache,
costan2016intel, van2018foreshadow,
lee2017inferring, xu2018controlled,
brasser2017software}, although these leaks can be 
mitigated by making code
data-oblivious~\cite{ohrimenko2016oblivious,hynes2018efficient,
tople2018privado, rane2015raccoon}, and
formally proving the absence of digital side-channels and bugs~\cite{bond2017vale}. 
Second, TEEs can leak data through analog side-channels
such as power draw, and physical side-channels such as bus
tapping~\cite{genkin2015get, genkin2015stealing,
kocher1999differential}.
Third, current systems use a single TEE which opens up the possibility
that the hardware designer or someone along the supply chain injects a
backdoor into the hardware~\cite{walen2016intel, bighackbloomberg,
ciscobackdoor, militarybackdoor}. 
Thus, current TEE-based systems also rely on the assumption that the
TEE vendor is trustworthy.

We introduce \sys, a new two-server ML inference system in the honest-but-curious
model. \Sys uses cryptography but runs a piece of functionality inside TEEs
to remove some weight from cryptography, giving a
substantially less expensive system that solely using cryptography. 
The use of 
TEEs does create a trusted computing base (TCB), consisting of both 
    the code that runs inside the TEEs, and
the hardware design and implementation of the TEE itself.
However, \sys lowers the size of the TCB, by (i) reducing
the size of the functionality running inside the TEE, (ii) 
distributing trust over heterogeneous TEEs from different vendors such
that the confidentiality of the system is preserved even if a TEE is
compromised. 

\Sys performs (i) and (ii) above for reducing TCB size in two progressive design steps
that we call \emph{reducing-TEE-code} and \emph{distributing-trust}. 

\emparagraph{Techniques for reducing-TEE-code step.} 
\Sys starts by observing that ML inference computation for many types of models,
particularly,  neural networks, can be expressed as a 
    series of layers, where each layer performs \changebars{either a linear computation (a
vector-matrix product) or a non-linear computation (an activation function such
as Rectified Linear Unit or ReLU), or both}{a linear computation (a
vector-matrix product) and a non-linear computation (an activation function such
as Rectified Linear Unit or ReLU)} (\S\ref{s:problem}). 
    Thus, a secure solution for ML inference requires sub-protocols for linear
and non-linear computations.

A common way to perform these computations over two servers is to use the Beaver
multiplication protocol to compute vector-matrix
products~\cite{beaver1991efficient} and Yao's
garbled circuit protocol~\cite{yao1982protocols} to compute non-linear functions. Further, Beaver's
protocol requires the two servers to hold correlated randomness called Beaver
triple shares, which is typically generated using additive homomorphic
encryption~\cite{mohassel2017secureml, keller2018overdrive}. A challenge with existing protocols is that both Beaver triple generation 
    and Yao's protocol 
incur significant expense. For instance, Yao's protocol requires transferring 
a verbose Boolean circuit representation of the non-linear function between the
servers.

\Sys makes two changes to this protocol.
First, instead of using additive homomorphic encryption to generate Beaver
triple shares, it uses a new protocol based on homomorphic secret sharing or
HSS~\cite{boyle2016breaking,boyle2019homomorphic,boyle2018foundations}
(\S\ref{s:triplegeneration}).
This protocol contains a packing technique that optimally uses the input space
of HSS operations, thereby reducing overhead relative to additive 
    homomorphic encryption-based solutions. For instance, network overhead in \sys's
protocol for a vector-matrix product over a matrix with $1024 \times 1024$
entries is at least 7.6$\times$ lower relative to prior work
(\S\ref{s:eval:overheadlinear}).

Second, 
    \sys replaces Yao's general-purpose protocol with a 
    recent protocol of Boyle, Gilboa, and Ishai (BGI)~\cite{boyle2019secure}
        that is efficient for computing non-linear functions.
    This protocol consumes slightly more \cpu than Yao, but incurs significantly lower
        network overhead (for instance, by 460$\times$ 
        for the ReLU function; \S\ref{s:eval:overheadnonlinear}).
    
The BGI protocol is promising; however, applying it to ML inference creates two
issues. First, the protocol, as described in the literature can efficiently encode
    the ReLU function (and approximations of Sigmoid and
Tanh) but not the MaxPool and Argmax functions~\cite{boyle2019secure,
ryffel2020ariann} (\S\ref{s:beaver-and-bgi}). The core issue is that the BGI protocol 
        depends on the function secret sharing (FSS)
primitive~\cite{boyle2015function, boyle2016function}, whose
        current constructions exist only for two functions, a point 
        function and an interval functions, that do not naturally
        encode the max function. The second issue with the BGI protocol is that it assumes
        that the two servers hold correlated randomness: shares
            of keys for FSS.

    \Sys fixes the first issue via new encodings of MaxPool and Argmax atop
    point and interval functions (\S\ref{s:beaver-and-bgi}). 
    These encodings may be of independent interest.
    \sys fixes the second issue by generating the FSS keys inside a 
        TEE per server.
    Since all non-linear functions further call just the point and interval
        functions, the code for FSS key generation is small.

\emparagraph{Techniques for distributing-trust step.}
    The protocol so far is efficient but contains TEEs as a single point
        of attack. In the distributing-trust step, 
        \sys takes FSS key generation
            from inside the TEEs and distributes it over \emph{multiple,
                heterogeneous} TEEs from different vendors.
    In particular, \sys replaces one TEE per server with three TEEs per server
        and runs a three-party secure computation protocol (3PC) over the TEEs
            such that an adversary does not learn the FSS keys even if
                it corrupts one of the TEEs.

General-purpose 3PC protocols can be expensive. 
    However, \sys again notes that
        all ML non-linear functions can be encoded
    on top of the point and interval functions. 
    So it devises a new customized 3PC protocol for the limited
    functionality of generating FSS keys for point and interval functions. 
        \sys's customized protocol is 
            cheaper, for instance, by 
            30$\times$ in terms of network transfers, relative to a general
        solution (\S\ref{s:fsskeygeneration}).

\emparagraph{Evaluation results.}
    We have implemented (\S\ref{s:impl}) and evaluated (\S\ref{s:eval}) 
        a prototype of \sys. Our prototype runs
            over two cloud providers, Microsoft Azure and 
        Amazon AWS, with multiple TEE machines per provider. 
    Our prototype demonstrates two properties of \sys.
        First, its code inside the TEE is 14.6--29.2$\times$
            smaller relative to existing single TEE-based systems (in absolute
                 terms, it is less than 1,300 lines of code). 
        Second, for several ML models
            including a 32-layer ResNet-32~\cite{he2016deep}, and for several datasets including
                those for speech and image recognition, 
    \sys's dollar cost to perform inference (that is, \cpu and network
consumption converted to a dollar amount) is
                5.4--385$\times$ lower than prior state-of-the-art
            cryptography-based works that run over two non-colluding servers.
\section{Overview of \sys}
\label{s:background-and-overview}

\subsection{Private outsourced ML inference}
\label{s:problem}

        \Sys targets the problem of private outsourced ML inference. This
        problem revolves
        around three parties: an \emph{ML model owner}, a \emph{service provider}, 
        and a \emph{data point owner}.
	    The model owner trains a model and deploys it 
                at the service provider,
            whose task is 
            to label new data points supplied by the data point owner against the model, for example,
        tell whether an image contains a human or not.
        The privacy aspect of the problem requires that 
        (a)~the service provider must not learn the model parameters or the data
        points, (b)~the model owner must not learn the data points, and (c)~the
        data point owner must learn only the inference result.
\begin{figure}[t!]
  \centering{\includegraphics[width=3.35in]{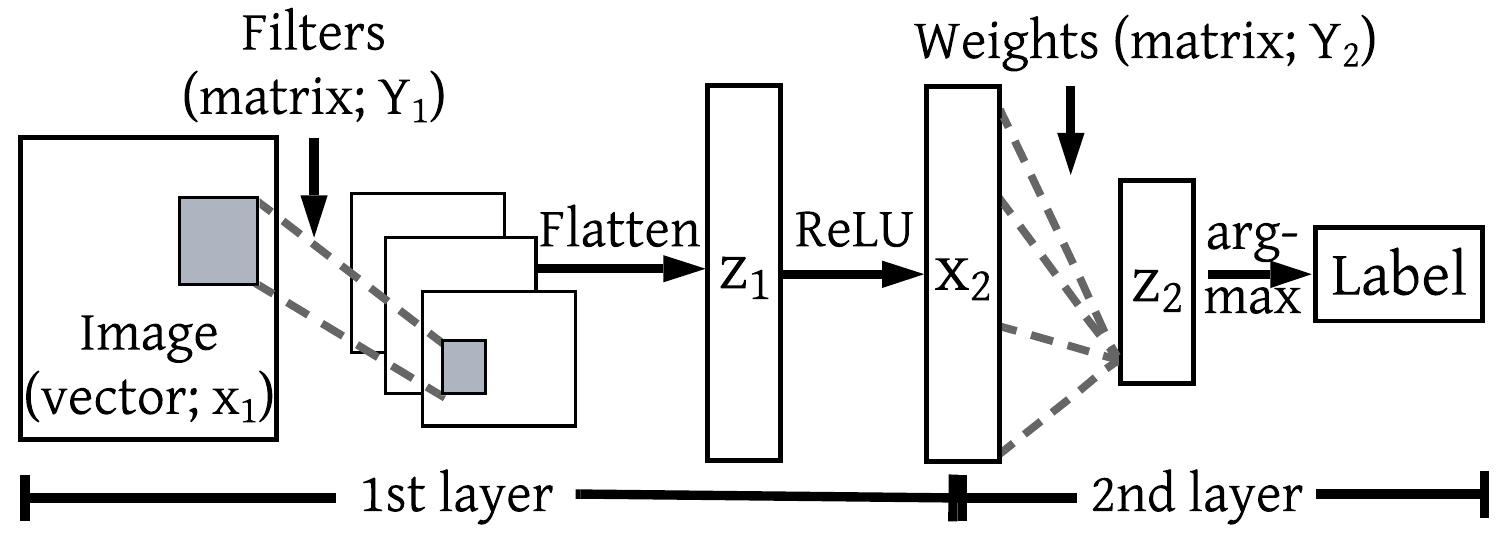}}
\caption{An example CNN with two layers. Each
layer computes a vector-matrix product (e.g., $\myvec{z}_1 = \myvec{x}_1
\cdot \mymatrix{Y}_1$) and applies a non-linear function (e.g., ReLU).}
\label{f:cnn}
\label{fig:cnn}
\end{figure}

While many types of ML models exist, \Sys focuses on neural
networks~\cite{haykin1994neural,goodfellow2016deep}, in particular, feedforward neural networks
(FNNs) and convolutional neural networks (CNNs), for two reasons. First, FNNs and CNNs have a wide
array of applications, from speech recognition~\cite{abdel2014convolutional}, to computer
vision~\cite{krizhevsky2012imagenet}, to chemical analysis~\cite{svozil1997introduction}.  Second,
one can express inference for other models such as support vector machines, Naive Bayes, and
regression as inference over FNNs~\cite{mohassel2017secureml}.

Fundamentally, FNNs and CNNs rely on slightly different building blocks. For instance, the former
employs dense layers while the latter additionally uses convolutions. However, one can abstract
inference for both types of models into a common structure. This structure is a series of
\emph{layers}, where each layer computes a vector-matrix product and applies a non-linear function such as ReLU (Rectified Linear Unit), Sigmoid, MaxPool, Argmax, and Tanh to the
product~\cite{goodfellow2016deep}. Figure~\ref{f:cnn} illustrates the structure of an example CNN.

\begin{figure}[t!]
  \centering{\includegraphics[width=3.35in]{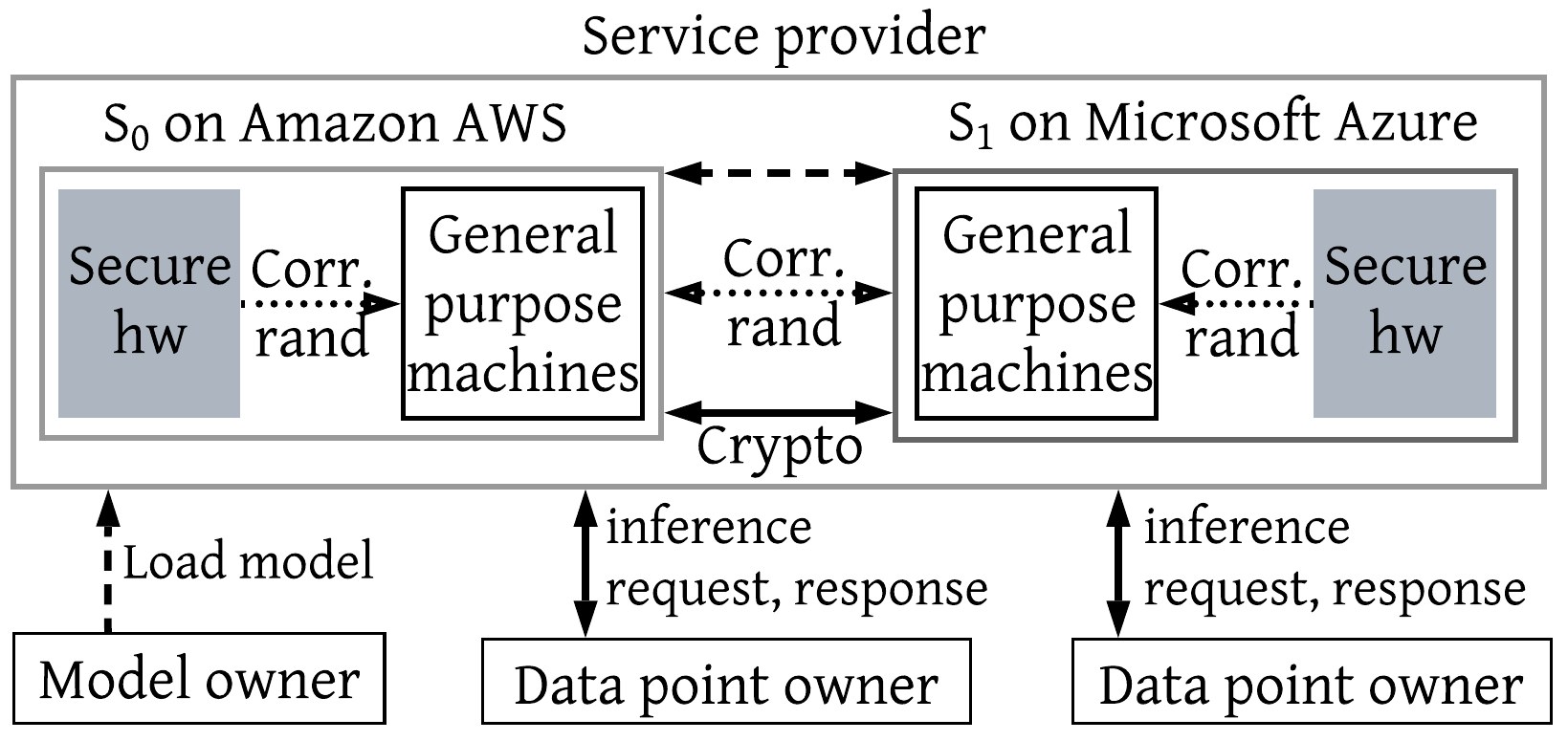}}
\caption{\Sys's high-level architecture.}
\label{f:arch}
\label{fig:arch}
\end{figure}

\subsection{Architecture}
\label{s:arch}

Figure~\ref{fig:arch} shows \sys's 
architecture. 
\Sys consists of a service provider, and multiple model and data point
owners. The service provider runs two servers, $S_0$ and $S_1$, in separate administrative domains
such as Microsoft Azure and Amazon AWS. Each server contains TEE machines from different vendors (labeled collectively as ``secure hw'' in
the figure), and general-purpose compute machines (labeled collectively as ``general-purpose
machines'' in the figure). 

 At a high level, \sys's protocol to privately outsource inference
        has four phases: \emph{setup},  
            \emph{model-loading},
            \emph{preprocessing}, and 
            \emph{online}.

\begin{myitemize2}

\item {\bf Setup}: The setup phase (not depicted in Figure~\ref{f:arch}) runs
once between \changebars{the}{} \sys's two
servers. 
During setup, the servers generate long-lived cryptographic material such
as seeds for a pseudorandom number generator. This cryptographic material is reused across all
inference requests.
                
\item {\bf Model-loading}: This phase (dashed arrows in Fig.~\ref{f:arch}) runs once per model. In
this phase, the model owner uploads \emph{secret-shares} (over a field) of the model parameters to
the two servers, who transform and store them. 
We denote $sh_b^{MP}$ as the share given by model owner to $S_b$, for $b \in \{0,1\}$.

\item {\bf Preprocessing}: The preprocessing phase runs once per inference, and precedes the online
phase. In the preprocessing phase (depicted by dotted arrows in
Figure~\ref{f:arch}), \sys's servers
generate \emph{correlated randomness} (depicted as ``corr rand'' in the figure). This
correlated randomness does not depend on the values of the model parameters or
the data points. This phase uses the TEE machines. 

        \item {\bf Online}: The online phase (depicted by solid arrows in Figure~\ref{f:arch}) runs
            once per inference and is input-dependent. In this phase, the data point owner uploads
            secret-shares (over a field) of its data point to the two servers (we denote
            the share $sh_b^{DP}$ as the share given to $S_b$), who run a cryptographic protocol
            using these shares and the outputs from the other phases, and generate shares of the
            inference label. Finally, each server sends
        its share of the label to the data point owner, who 
        combines the shares to get the actual label.

        \end{myitemize2}

\begin{definition}[Correctness]
For every input $sh_b^{MP} \in \{0,1\}^{\mathrm{poly}(\lambda)}, sh_b^{DP} \in
\{0,1\}^{\mathrm{poly}(\lambda)}$ (corresponding to 
    a data-point $DP$ sent by the
data-point owner, and
additive secret-shares of
model parameters $MP = \{\mymatrix{Y}_{0}, \mymatrix{Y}_{1},..., \mymatrix{Y}_{L
- 1}\}$ for a $L$-layer model sent by the model-owner)
to $S_b$, the reconstruction of the secret-shares output by
both servers $S_0, S_1$ equals the prediction output of applying the 
model with parameters $MP$ to input $DP$.
\end{definition}

\subsection{Threat model and security definitions}
\label{s:threat-model}
\label{s:security-definition}
\label{s:security-definitions}
\Sys considers an honest-but-curious adversary. This adversary 
follows the description of the protocol but tries to infer 
sensitive data by inspecting protocol messages. 
Below, we formally define \sys's security notion.

We consider two settings, namely \emph{single-TEE} and \emph{multiple-TEE}, 
depending on how many TEEs \sys's servers
employ. 
In the
\emph{single-TEE} setting, \sys's servers use one TEE each.\footnote{The TEE is
logically centralized but may be distributed over many physical TEE
machines of the same type.} We denote
the TEE used by $S_b$ as $T_b$, for $b \in \{0,1\}$. Further, we denote 
the functionality implemented by $T_b$ as $\func_b$. Finally, we denote non-TEE
machines at $S_b$ collectively as $M_b$.
In the {\em multiple-TEE} setting, each server
uses three types of TEEs (made by three different vendors); we denote
$S_b$'s three types of TEEs by $T_b^{(i)}$ for $i \in \{0, 1, 2\}$ such that
$i$-th TEE $T_b^{(i)}$ implements functionality $\func_b^{(i)}$. We define
security for the two settings separately.

\begin{definition}[Single-TEE security]
\label{def:single-tee}
A single-TEE \sys scheme consisting of setup, model-loading, preprocessing, and
online phases is said to be $\varepsilon$-secure if for any honest-but-curious
(passive) probabilistic polynomial time (PPT) adversary $\adversary$ corrupting
$M_b$ for $b \in \{0,1\}$ with access to TEE $T_b$ implementing $\func_b$, for
every large enough security parameter $\lambda$, there exists a PPT simulator
$\simr$ such that the following holds: 

for all inputs $sh_b^{MP} \in \{0,1\}^{\mathrm{poly}(\lambda)},sh_b^{DP} \in
\{0,1\}^{\mathrm{poly}(\lambda)}$ to $S_b$, randomness $r_b \in
\{0,1\}^{\mathrm{poly}(\lambda)}$, 
$$\{\view_{\adversary}^{\func_b}(1^{\lambda},sh_{b}^{MP},sh_b^{DP};r_b) \}
\approx_{c,\varepsilon}$$ $$\{\simr(1^{\lambda},sh_b^{MP},sh_b^{DP},r_b)\}.$$ 

If $\varepsilon$ is negligible in the security parameter, we drop
$\varepsilon$ in the above definition.  
\end{definition}

\begin{definition}[Multiple-TEE security]
\label{def:multiple-tee}
A multiple-TEE \sys scheme consisting of setup, model-loading, preprocessing,
and online phases is said to be $\varepsilon$-secure if, for any
honest-but-curious (passive) probabilistic polynomial time (PPT) adversary
$\adversary$ corrupting $M_b, T_b^{(i)}, T_{1 - b}^{(i)}$ for $b \in \{0,1\}, i
\in \{0, 1, 2\}$, with access to TEEs $T_b^{(j)}, T_b^{(k)}$ implementing
$\func_b^{(j)}, \func_b^{(k)}$ respectively, for $j,k \in \{0,1,2\}$ and $j \neq
i,k \neq i$, for every large enough security parameter $\lambda$, there exists a
PPT simulator $\simr$ such that the following holds:

for all inputs $sh_b^{MP} \in \{0,1\}^{\mathrm{poly}(\lambda)},sh_b^{DP} \in
\{0,1\}^{\mathrm{poly}(\lambda)}$ to $S_b$, randomness $r_b \in
\{0,1\}^{\mathrm{poly}(\lambda)}$, 
$$\{\view_{\adversary}^{\func_b^{(j)},\func_b^{(k)}}(1^{\lambda},sh_{b}^{MP},sh_b^{DP};
r_b)\} \approx_{c,\varepsilon}$$
$$\{\simr(1^{\lambda},sh_b^{MP},sh_b^{DP},r_b)\}.$$ 

If $\varepsilon$ is negligible in the security parameter, we drop
$\varepsilon$ in the above definition.

\end{definition}

\emparagraph{Remark.} We assume that TEEs $T_b^{(i)}, T_{1 - b}^{(i)}$ come from the same manufacturer.

\noindent We do not consider attacks such as membership
inference~\cite{shokri2017membership} and
model stealing~\cite{tramer2016stealing} that aim to infer membership in
training dataset or learn approximate
model parameters by observing the black-box
behavior of the ML inference system. Although this leakage is an important concern,
secure computation alone cannot prevent it. However, defending against such
attacks is an active area of
research~\cite{juuti2019prada,jia2019memguard,orekondy2019prediction}. 
    Besides, 
these attacks are 
immaterial when the entity receiving inference outputs also owns
the model (that is, when a model owner remotely deploys a model for its own
consumption).

\subsection{Prior approaches and related work}
\label{s:relwork}
\label{s:approaches}

       Several approaches exist in the literature for privately outsourcing the 
        task of inference over FNNs and CNNs. 
       Here, we compare \sys with 
                these prior approaches.
        While doing the comparison, we include prior works for a restricted
            setting 
            where the service provider has access to model parameters in
            plaintext,  
        as the techniques developed for this
            restricted setting are related to the techniques in \sys's
                fully-outsourced setting that also hides model
parameters. 

        One can split prior works 
        into two broad categories: those that rely on TEEs for their security
guarantees and those that rely only on cryptography.

        \emparagraph{TEE-based works.} 
        The works based on TEEs use the popular Intel SGX TEE~\cite{ohrimenko2016oblivious, hynes2018efficient,
        tople2018privado, hunt2018ryoan, hunt2018chiron,
            tramer2018slalom, narra2019privacy,fischer2020computation}.
        Many of these works~\cite{ohrimenko2016oblivious, hynes2018efficient,
        tople2018privado, hunt2018ryoan, hunt2018chiron}
            run a complete ML system inside
            the TEE. This approach is efficient as the code
            runs natively on the \cpu. However, 
        as indicated earlier (\S\ref{s:introduction}),
        systems based on a single, general-purpose TEE are
            vulnerable to many attacks.

        Slalom~\cite{tramer2018slalom}, Origami~\cite{narra2019privacy},
and DFAuth~\cite{fischer2020computation} also use Intel SGX, but, like \sys,
        move parts of inference outside of the TEE.
        However, these prior systems use the TEE during the online phase of inference,
            while \sys restricts TEE use to a preprocessing phase. Moreover, \sys 
            removes TEE as a single point of failure by securely
            distributing trust over heterogeneous TEEs.
        (Note that, unlike \sys, Slalom and Origami do not hide model parameters from the
            service provider.)

        \emparagraph{Cryptography-based works.} The alternative approach to using TEEs 
            is to use cryptographic constructs. In particular, a long line of works focuses on 
        building secure ML inference either from secure multiparty computation
(MPC)~\cite{riazi2018chameleon,
wagh2019securenn, kumar2020cryptflow, wagh2020falcon, mohassel2018aby,
        mohassel2017secureml, riazi2019xonn,ball2019garbled,liu2017oblivious,
                        juvekar2018gazelle,rouhani2018deepsecure,mishra2020delphi,
                    chen2019secure}, or fully homomorphic encryption (FHE)~\cite{gilad2016cryptonets,
xie2014crypto, hesamifard2017cryptodl, chabanne2017privacy,boemer2018ngraph,
badawi2018alexnet,brutzkus2019low, chou2018faster,
bourse2018fast,lou2019glyph, jiang2018secure,ryffel2020ariann}. 
        However, all these works incur higher overhead in comparison with TEE-based solutions.  
       For instance, a recent state-of-the-art system,
Glyph~\cite{lou2019glyph}, based on FHE, requires 
            $2^n$
                homomorphic operations for a non-linear function over a $n$-bit input. 
    \Sys focuses on the two-server secure computation (2PC)
setting; its \cpu and network overhead, when converted to dollars, is
5.4-385$\times$ lower than prior 2PC works
for this setting (\S\ref{s:eval:overheadinference}). 
    One can say that \sys's use of TEEs helps accelerate cryptography.

\subsection{Design approach}
\label{s:design-steps}

\begin{figure}[t!]
  \centering{\includegraphics[width=3.35in]{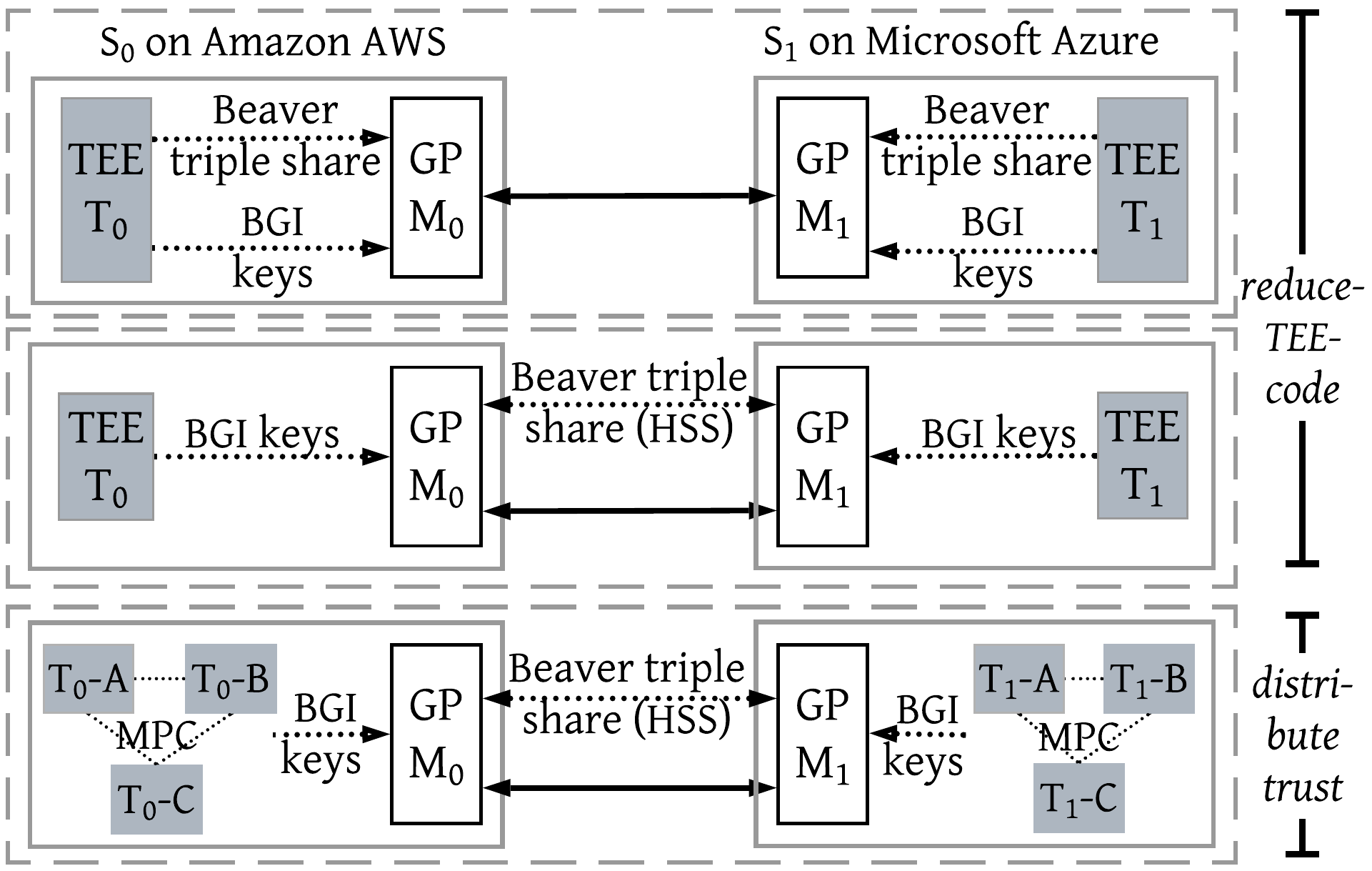}}
\caption{\Sys's design steps.
Dotted and solid arrows respectively show computation performed during
preprocessing and online phases of inference.}
\label{f:design-steps}
\label{fig:design-steps}
\end{figure}

As stated in the introduction (\S\ref{s:introduction}), \sys adopts the
two-step approach of {\em reducing-TEE-code} and {\em distributing-trust} for its
design.
Figure~\ref{f:design-steps} depicts these two steps.

At a high level, \sys starts with a solution that runs ML inference (all four
phases) 
inside a single TEE.
It then, gradually, via the reducing-TEE-code, moves most of the computation,
particularly, the frequently invoked preprocessing and online phases,
    from inside the TEE to outside the TEE. This step is further divided into two
sub-steps: the first sub-step completely gets rid of TEE in the online
    phase, and the second sub-step splits the preprocessing phase such that the
        bulk of preprocessing also happens outside the TEE. As a result, \sys greatly
simplifies the computation being performed inside the TEE. Finally, to avoid a single point
of attack, \sys employs a secure computation protocol over the
computation inside the
TEEs to distribute
        trust
among multiple, heterogeneous TEEs.
                
\emparagraph{First part of reducing-TEE-code: online phase without TEE.}
        In this first sub-step of reducing-TEE-code (first row in
Figure~\ref{f:design-steps}), 
            the TEEs at the two servers run 
            the 
            complete
            preprocessing phase.
            In particular, they generate two types of 
        input-independent correlated randomness: \emph{Beaver triple
shares}~\cite{beaver1991efficient} and 
            keys for a cryptographic protocol due to Boyle, Gilboa, and
Ishai (BGI)~\cite{boyle2019secure}.
        Meanwhile, the non-TEE machines run the complete online phase (\S\ref{s:beaver-and-bgi}).
                
         \emparagraph{Second part of reducing-TEE-code: less hardware, more software
for the preprocessing phase.}
        In the second sub-step of reducing-TEE-code 
            (illustrated in the middle row in
                Figure~\ref{f:design-steps}),
        \sys moves 
            a major part of the preprocessing phase---the generation of Beaver
                triple shares---outside of TEEs. 
        To generate these shares efficiently, 
        the non-TEE machines at \sys's two servers run
            an optimized cryptographic protocol based on a recent
        primitive called homomorphic secret sharing or
\emph{HSS}~\cite{boyle2016breaking,boyle2019homomorphic,boyle2018foundations}
        (\S\ref{s:triplegeneration}).
        After the 
        second sub-step of reducing-TEE-code,
               like a special-purpose cryptoprocessor~\cite{4769, ibmhsm, arthur2015practical},
        the TEE machines
            run the specialized task of generating keys for the BGI protocol.
                
    \emparagraph{Distributing-trust.}
            \Sys's distributing-trust step (illustrated in the bottom row in
            Figure~\ref{f:design-steps})
            reduces trust on TEEs, by 
            distributing the task of generating keys for BGI
            onto multiple TEEs. 
        To distribute key-generation efficiently, \sys uses
            a new, customized three-party secure computation protocol that achieves lower overhead (both \cpu consumption and network
            transfers) than 
            a general-purpose protocol, by shifting the computation of a
               \changebars{pseudorandom generator (PRG)}{pseudorandom number generator} (which is the bulk of the
                    computation in the key generation procedure) outside of the
general-purpose protocol (\S\ref{s:fsskeygeneration}).

The next three sections (\S\ref{s:beaver-and-bgi}, \S\ref{s:triplegeneration},
\S\ref{s:fsskeygeneration}) dwell exhaustively on the details of these design steps.
\section{Details of first part of reducing-TEE-code}
\label{s:beaver-and-bgi}
\begin{figure*}[t]
\hrule
\medskip

\begin{center}
\textbf{\Sys's protocol for first part of its reducing-TEE-code step}
\vspace{-2mm}
\end{center}

\begin{myitemize5}

\item This protocol has two parties, $S_0$ and $S_1$. It computes 
            shares of $f(\myvec{x} \cdot \mymatrix{Y})$, where vector
$\myvec{x}$ is in 
        $\mathbb{Z}_p^{1 \times n}$, matrix $\mymatrix{Y}$ is in 
$\mathbb{Z}_p^{n \times m}$, and $f$ is a non-linear function. We denote $sh_b^{(\myvec{x})}$ and $sh_b^{(\mymatrix{Y})}$ to be $S_b$'s shares of $\myvec{x}$ and $\mymatrix{Y}$ respectively.  

\item The protocol assumes that $S_b$ has a TEE machine $T_b$ and a
general-purpose machine $M_b$. It also assumes several cryptographic
primitives, as described below.
\end{myitemize5}

\begin{center}
\vspace{-4mm}
\textit{Setup phase}
\vspace{-2mm}
\end{center}

        \begin{myenumerate4}
        \item \label{e:diffiehellman} $T_0, T_1$ 
            establish a common seed for a pseudorandom function
            using the Diffie-Hellman protocol~\cite{diffie1976new, chevalier2009optimal}.
        \end{myenumerate4}

\begin{center}
\vspace{-2mm}
\textit{Model-loading phase}
\vspace{-1mm}
\end{center}

        \begin{myenumerate4}
        \setcounter{enumi}{1}
        \item \label{e:sampleB} $T_b$ samples $\mymatrix{B} \in_R \mathbb{Z}_p^{n
        \times m}$ and outputs its share $sh_b^{(\mymatrix{B})}$ to $M_b$.

        \item \label{e:receivemodelparams} (Receive model parameters $\mymatrix{Y}$) $M_0$ and $M_1$ respectively receive
                $sh_0^{(\mymatrix{Y})}$ and  
                $sh_1^{(\mymatrix{Y})}$ from the model owner.

        \item \label{e:maskmodelparams} (Mask $\mymatrix{Y}$) $M_0$ and $M_1$
                obtain 
                    $\mymatrix{F} = \mymatrix{Y} - \mymatrix{B}$, which is a
                    masked version of $\mymatrix{Y}$. 
                To obtain $\mymatrix{F}$, $S_b$ 
            computes 
                $sh_b^{(\mymatrix{F})} = 
            sh_b^{(\mymatrix{Y})} -
        sh_b^{(\mymatrix{B})}$, sends $sh_b^{(\mymatrix{F})}$ to $S_{1-b}$,
receives $sh_{1-b}^{(\mymatrix{F})}$ from $S_{1-b}$,  
            and computes
            $\mymatrix{F} = 
            sh_0^{(\mymatrix{F})} + 
            sh_{1}^{(\mymatrix{F})}$.
        \end{myenumerate4}

\begin{center}
\vspace{-2mm}
\textit{Preprocessing phase}
\vspace{-1mm}
\end{center}

        \begin{myenumerate4}
        \setcounter{enumi}{4}
            \item (Generate Beaver triple shares)
\label{e:beavertriplegeneration}
           $T_b$ samples $\myvec{a} \in_R \mathbb{Z}_p^{1
            \times n}$ and computes $\myvec{c} = \myvec{a} \cdot \mymatrix{B}$.
            It
                gives
                the Beaver triple share ($sh_b^{\myvec{a}}, sh_b^{\mymatrix{B}},
            sh_b^{\myvec{c}}$) to $M_b$.

            \item (Generate FSS keys) \label{e:fsskeygeneration}
                $T_b$ samples $\myvec{r} \in_R \mathbb{Z}_p^{1 \times m}$ and
                            outputs its share $sh_b^{(\myvec{r})}$ to $M_b$.
                                $T_b$ also computes FSS keys, $k_0$ and $k_1$, such
                                        that
                                            $(k_0, k_1) \leftarrow \fssgen\left( 1^{\lambda}, \widehat{f}_{\myvec{r}} \right)$,
                                            where $\widehat{f}_{\myvec{r}}(\myvec{in}) =
                                            f(\myvec{in} - \myvec{r})$ is an offset
                                            function for $f$. $T_b$ outputs key $k_b$ to $M_b$.
        \end{myenumerate4}

\begin{center}
\vspace{-2mm}
\textit{Online phase}
\vspace{-2mm}
\end{center}

    \begin{myenumerate4}
        \setcounter{enumi}{6}
        \item \label{e:receivedatapoint}
        $M_b$ receives $sh_b^{(\myvec{x})}$ from the
            data point owner (or from the output of step~\ref{e:fssonline}). 

        \item (Beaver multiplication)  \label{e:beavermultiplication}
$M_b$ takes matrix $\mymatrix{F}$ from the model-loading phase, Beaver
            triple share ($sh_b^{\myvec{a}}, sh_b^{\mymatrix{B}},
            sh_b^{\myvec{c}}$) 
            from the preprocessing phase, $sh_b^{(\myvec{x})}$ from the above
step,
    and 
            performs Beaver
        multiplication~\cite{beaver1991efficient}. $M_b$ obtains the output $sh_b^{(\myvec{x} \cdot \mymatrix{Y})}$. 

        \item \label{e:fssonline} 
        (BGI evaluation) \label{e:fsseval} $M_b$ takes its share of 
            $\myvec{x} \cdot \mymatrix{Y}$ from the above step, and
                FSS key $k_b$ 
                and randomness $sh_b^{(\myvec{r})}$ 
                    from
                the preprocessing phase, 
                and outputs 
                    $sh_b^{(f(\myvec{x} \cdot
        \mymatrix{Y}))}$ 
                using the BGI protocol~\cite{boyle2019secure}. 
    \end{myenumerate4}
\hrule
\caption{This protocol composes the Beaver multiplication protocol for computing vector-matrix
products~\cite{beaver1991efficient} with the function secret sharing (FSS)-based 
    BGI protocol for computing non-linear
functions~\cite{boyle2019secure}. The two sub-protocols require correlated
randomness, which is generated using TEE machines during the preprocessing phase.}
\label{fig:beaver-and-bgi}
\label{f:beaver-and-bgi}
\end{figure*}

This section describes \sys's protocol for the first part of its
reducing-TEE-code step. 
To begin with, we focus on one layer of inference, that is, computing one
vector-matrix product and applying a non-linear function to the output
of the product (\S\ref{s:problem}); 
later in this section, we will relax this assumption. 

Figure~\ref{f:beaver-and-bgi} shows \sys's protocol for one layer of inference. 
Say that the vector is $\myvec{x} \in \mathbb{Z}_p^{1 \times n}$ and the matrix is
$\mymatrix{Y} \in \mathbb{Z}_p^{n \times m}$, then the protocol computes
$f(\myvec{x} \cdot \mymatrix{Y}) \in \mathbb{Z}_p^{1 \times m'}$, where $m' \leq
m$, and $f$ is
the non-linear function such as ReLU or MaxPool. The vector $\myvec{x}$
is the data point from the data point owner or the output of the previous
layer; the matrix $\mymatrix{Y}$ encodes model parameters. All arithmetic
is in the 
field 
$\mathbb{Z}_p$ for a prime $p$.

Underneath, the protocol composes Beaver's secure multiplication protocol~\cite{beaver1991efficient} with a 
    protocol due to Boyle, Gilboa, and
    Ishai (BGI)~\cite{boyle2019secure}. 
The Beaver part securely computes the vector-matrix product: it takes as input
the shares of the vector $\myvec{x}$ and matrix $\mymatrix{Y}$, and the shares of a
\emph{Beaver triple} 
($\myvec{a}, \mymatrix{B}, \myvec{c}$), and
generates shares of the vector-matrix product $\myvec{z} = \myvec{x} \cdot
\mymatrix{Y}$. 
    For an unfamiliar reader, the Beaver triple 
($\myvec{a}, \mymatrix{B}, \myvec{c}$) is a vector-matrix product over a random vector
and matrix. That is, 
    $\myvec{a}$ and $\mymatrix{B}$ are sampled uniformly at random with elements
in $\mathbb{Z}_p$, $\textrm{dim}(\myvec{a}) = \textrm{dim}(\myvec{x})$, 
    $\textrm{dim}(\myvec{B}) = \textrm{dim}(\myvec{Y})$, 
    and $\myvec{c} = \myvec{a} \cdot \myvec{B}$.

The BGI part of the protocol computes the non-linear
function: it starts with the shares of the vector-matrix product $\myvec{z}$, and the
shares of
the non-linear function $f$, 
    and computes the shares of the non-linear function applied to the product,
that is, shares of 
$f(\myvec{z})$.
A key enabler of the BGI protocol is the function secret sharing
(FSS) primitive (\fssgen,
\fsseval)~\cite{boyle2015function, boyle2016function}. \fssgen
splits a function $f$ into two secret shares $f_0$ and $f_1$, called FSS keys, such that 
$f_0(x) + f_1(x) = f(x)$. 
\fsseval evaluates a
share $f_b$, for $b \in \{0, 1\}$,  on an input $x$ to give a
share of $f(x)$ (this happens in step~\ref{e:fsseval} in
Figure~\ref{fig:beaver-and-bgi}).

The Beaver triples and FSS keys form input-independent correlated randomness. 
The protocol uses the TEE machines to generate this randomness during
the preprocessing phase.

\emparagraph{Supporting multiple layers of inference.}
The protocol above works for one layer of inference. To support
multiple layers, \sys replicates the computation inside each phase
(except the setup phase) as many times as the number of layers.
It then connects copies of the online phase for adjacent layers. 
Specifically, it feeds the output of step~\ref{e:fsseval}, which is a vector in
$\mathbb{Z}_p$, to step~\ref{e:receivedatapoint}, which expects a vector of the
same type.

\emparagraph{Lack of expressibility and fixes.} 
There are two issues with expressibility of the described protocol.
First, it assumes arithmetic over the field $\mathbb{Z}_p$, whereas 
neural networks perform arithmetic over floating-point numbers. 
\Sys addresses this issue by borrowing standard techniques from the literature
to encode floating-point arithmetic as field arithmetic~\cite{mishra2020delphi, mohassel2017secureml}. 
The conversion results in a drop in inference accuracy; however, this
drop is small (\S\ref{s:eval:overheadinference}).

The second issue with expressibility is that the BGI part of the protocol can
directly handle only certain ML non-linear functions.
The restriction is due to the fact that efficient FSS constructions
currently exist only for 
two functions: a point function
$f_{\alpha}^{\beta}$ that outputs $\beta$ at the point $\alpha$ and zero
otherwise, and the interval function 
$f_{(\alpha_1, \alpha_2)}^{\beta}(x)$ 
    that outputs $\beta$ if $\alpha_1 \leq x
\leq \alpha_2$ and zero otherwise. 
These functions can express a piece-wise polynomial (that is, a spline)
function~\cite{boyle2019secure}, which in turn can encode
the ReLU function and 
several close approximations~\cite{amin1997piecewise} of Sigmoid and
Tanh. However, a spline cannot directly encode the MaxPool and 
Argmax functions.

Normally, one would express a max over two values as $\textrm{max}(x, y)
= \textrm{sign}(x-y) \cdot (x-y) + y$, where the sign function (which is a
spline) returns 1 if its
input is positive and zero otherwise. 
However, this formulation
    of max does not work when the inputs $x, y$ are in 
$\mathbb{Z}_p$. For
example, 
consider the case where $x=5$, $y=3$, and $p = 7$.
For this
case, $y > x$ ($x \geq 4$ is considered negative) but
$\textrm{sign}(x-y)=1$ (+ve).
The problem is that 
$\mathbb{Z}_p$ 
(when it encodes both positive and negative numbers) is not 
a totally ordered set. 

There are many details to how \sys 
    encodes Maxpool and Argmax as a composition of point
    and interval functions; we leave these details to Appendices~\ref{a:maxpool}
and~\ref{a:argmax}.
However, \sys's key idea is to split the computation into two parts: when both 
    $x$ and $y$ have the same sign, and when they do not.
For the former case, that is, when both $x$ and
$y$ are either both positive or both
negative, 
    $\textrm{sign}(x-y)$ gives the right answer. 
Therefore, one can write $\textrm{max}(x,y)=\textrm{ReLU}(x-y) + y$. 
    For the case when $x$ and $y$ have different signs, one can write
$\textrm{max}(x,y)=\textrm{ReLU}(x)+\textrm{ReLU}(y)$. 
\Sys composes these two cases, again by using just point and interval 
    functions.

We note that Ryffel et al. in 
parallel work also encode MaxPool
and Argmax using point and interval functions~\cite{ryffel2020ariann}.
However, their protocol assumes a trusted third party (besides the two servers).
Furthermore, their encoding limits the inputs to a small subset of $\mathbb{Z}_p$, 
and incurs network overhead that is quadratic in the number of
input entries 
to MaxPool and Argmax. 
In contrast, \sys's encoding does not have an input restriction, and 
incurs network overhead linear in the number of input entries
to MaxPool and Argmax.

\emparagraph{Cost analysis.}
The cost of setup and model-loading phases in Figure~\ref{f:beaver-and-bgi} 
    gets amortized across inference
requests as model parameters typically change infrequently. Here, we discuss network and
\cpu costs for the preprocessing and online phases. 

\emph{Network overhead.}
In terms of network, the preprocessing phase requires the TEE machines to
transfer correlated randomness (Beaver triple shares and FSS keys) to
general-purpose machines. 
These data transfers are within a single administrative domain and cheap. 
Indeed, popular cloud providers do not charge for intra-domain
        transfers within a geographical zone~\cite{googlenetworkpricing,
azurenetworkpricing}. The online phase incurs inter-server (wide-area) network overhead
equal to the size of $\myvec{x}$ plus a small multiple of the size of 
$\myvec{z} = \myvec{x} \cdot \mymatrix{Y}$.
The first term is due to the 
Beaver part (step~\ref{e:beavermultiplication} in
Figure~\ref{fig:beaver-and-bgi}), while the second term is due to the BGI evaluation part (step~\ref{e:fsseval} in 
Figure~\ref{fig:beaver-and-bgi}). 
Note that, in contrast, prior work that uses 2PC between two servers
(\S\ref{s:relwork}) uses Yao's garbled
circuits~\cite{yao1982protocols} for non-linear functions, whose
network overhead is much higher---a multiple of the verbose Boolean circuit
representation of the non-linear function.
For instance, for ReLU, \sys's implementation of BGI costs
    18 bytes while a recent and optimized implementation of Yao~\cite{zahur2015obliv} costs
8.3~KB (\S\ref{s:eval:overheadnonlinear}).

\emph{\cpu overhead.} In terms of \cpu, the
Beaver part 
    computes vector-matrix products 
over small numbers in $\mathbb{Z}_p$ ($p$ is a 52-bit prime in
our implementation).
    Meanwhile, 
    the BGI part 
runs \fssgen and \fsseval over the point and interval functions. 
The \cpu for FSS procedures is higher than for Yao (for example, for ReLU,
1.3~ms versus
0.45~ms in Yao; \S\ref{s:eval:overheadlinear})
    as both \fssgen and \fsseval internally make
    many calls to AES (for example,
\fsseval for an interval function over a $p$-bit input performs $8 \cdot \log p$
AES encryptions). 
However, since \cpu is a much cheaper
resource than network consumption, the reduction in network
overhead outweighs the increase in \cpu.

\emparagraph{Security analysis.}
The protocol described in Figure~\ref{fig:beaver-and-bgi} satisfies the
single-TEE security definition in \S\ref{s:security-definitions}
(Appendix~\ref{a:single-tee-proof}).

\section{Details of second part of reducing-TEE-code}
\label{s:triplegeneration}

A limitation of the protocol in the previous section is the
high amount of computation performed by the TEE machines at the two servers
(steps~\ref{e:beavertriplegeneration} and~\ref{e:fsskeygeneration} in
Figure~\ref{fig:beaver-and-bgi}). In particular, the TEE machines generate FSS keys and 
Beaver triple shares.
Moreover, the latter requires a substantial amount of code inside the TEEs: 
    not only does the TEE compute vector-matrix products but 
        it also runs code to maintain state outside the TEE: 
        for each layer of every model, step~\ref{e:sampleB} in Figure~\ref{fig:beaver-and-bgi} samples
and stores a matrix $\mymatrix{B}$, and
step~\ref{e:beavertriplegeneration} reuses this state across inference requests
to generate triples. Therefore, \sys's second part of reducing-TEE-code step moves Beaver triple generation to general-purpose
(non-TEE) machines ($M_0$ and $M_1$). 

Observe that the first two components of a Beaver triple ($\myvec{a}, \myvec{B}, \myvec{c}$)
are sampled uniformly at random. Therefore, $M_b$
   can locally sample its shares
                            $sh_b^{(\myvec{a})}$ and
                                $sh_b^{(\myvec{B})}$ as the sums $\myvec{a} = 
                            sh_0^{(\myvec{a})} +
                            sh_1^{(\myvec{a})} \pmod p$ and 
                            $\mymatrix{B} = 
                            sh_0^{(\mymatrix{B})} +
                            sh_1^{(\mymatrix{B})} \pmod p$ 
                            are also uniformly random. To obtain shares of 
        $\myvec{c} = \myvec{a} \cdot \mymatrix{B}$
    from shares of $\myvec{a}$ and $\mymatrix{B}$, prior work offers several two-server 
protocols~\cite{mishra2020delphi, juvekar2018gazelle,
keller2018overdrive, keller2016mascot, mohassel2017secureml}.
However, these protocols incur a high inter-server (wide-area) network overhead.
    For instance, 
            for a vector with 128 entries and 
                a matrix with 128 $\times$ 128 entries,
        the network overhead
            of a additive homomorphic encryption-based 
        protocol used in the state-of-the-art prior
works~\cite{juvekar2018gazelle, mishra2020delphi, keller2018overdrive}
            is
        over $1{,}000$ times the size of the vector.

Instead of using prior homomorphic encryption-based protocols, \sys uses a new protocol 
based on a primitive called homomorphic 
secret sharing (HSS) that has received much attention
recently~\cite{boyle2016breaking, boyle2019homomorphic, boyle2018foundations,
fazio2017homomorphic}. \Sys's HSS-based 
    protocol significantly reduces (amortized) network overhead---for instance, to
$16\times$ the size of the vector for the specific example above.
However, obtaining this performance
requires addressing two challenges of applying HSS to Beaver triple
generation.
This 
section gives a necessary background on HSS, explains the challenges,
and describes \sys's protocol.

\subsection{Overview of Homomorphic secret sharing (HSS)}
\label{s:hss}

Homomorphic secret sharing or HSS~\cite{boyle2019homomorphic, boyle2016breaking,
boyle2018foundations} is a cryptographic
primitive that allows a client to outsource the computation of a program (containing addition and
multiplication instructions) to two non-colluding servers such that each server
produces its share of the program output without learning the original program
inputs.  

An HSS scheme has three procedures: $\hssgen$, $\hssenc$, and $\hsseval$.
To outsource a program $z = I(x, y, \ldots)$ over an input space
$\mathcal{I}$, a client first invokes 
$\hssgen$ to generate HSS keys. 
These keys consist of a public key, $pk$, for an underlying encryption scheme,
and the shares of the corresponding secret key, $(e_0 = sh_0^{(s)}, e_1 =
sh_1^{(s)})$. 
The client uses the public key to run
$\hssenc$ and produce
a set of ciphertexts, ${\bf C}$, containing encryptions of the program inputs $(x, y, \ldots)$. 
The client also produces two sets,
    ${\bf S_0} = \{sh_0^{(x \cdot s)}, sh_0^{(y \cdot s)}, \ldots\}$ 
        and 
    ${\bf S_1} = \{sh_1^{(x \cdot s)}, sh_1^{(y \cdot s)},
\ldots\}$,
                containing shares 
            of the program inputs times
    the secret key. 
The client sends 
$(pk, e_0, {\bf C}, {\bf S_0})$ to server $S_0$, and 
$(pk, e_1, {\bf C}, {\bf S_1})$ to server $S_1$.
Finally, server $S_b$ locally (without interaction with $S_{1-b}$) runs 
$\hsseval(e_b, {\bf C}, {\bf S_b}, I)$ 
and
gets its share of the program output $z$.

\Sys builds on \hssbkslprname = (\bkslprgen, \bkslprenc,
\bkslpreval)~\cite{boyle2019homomorphic} HSS scheme as it is the most efficient
HSS scheme
in the literature. There are three notable aspects of \hssbkslprname.
First, the input space $\mathcal{I}$ is the polynomial ring
$R_p = \mathbb{Z}_p[x]/(x^N+1)$ 
    consisting of all degree $N-1$ polynomials
   with coefficients in 
    $\mathbb{Z}_p$. 
Second, the underlying encryption scheme that \hssbkslprname uses is the LPR
scheme~\cite{lyubashevsky2010ideal} with plaintext space $R_p$.
Third, a key instruction that \bkslpreval supports is \emph{Mult}.
This instruction
takes as inputs 
a LPR ciphertext $C^x$ for $x \in R_p$, 
and a share of an input $y \in R_p$ times the LPR secret key, 
    that is, 
    a share of $y \cdot s$, and outputs
a share of the product $x \cdot y$.
That is, $sh_b^{(x \cdot y)} \gets \textrm{Mult}(sh_b^{(y \cdot s)}, C^{x})$.

\begin{figure*}[t]
\hrule
\medskip
{

\begin{center}
\textbf{\Sys's protocol after its reducing-TEE-code step}
\vspace{-3mm}
\end{center}

\begin{myitemize5}

\item This protocol assumes the same parties and performs the same computation as the
protocol in Figure~\ref{fig:beaver-and-bgi}. 

\end{myitemize5}

\begin{center}
\vspace{-3mm}
\textit{Setup phase}
\vspace{-3mm}
\end{center}

        \begin{myenumerate4}
        \item \label{e:hsskeygen} $M_0$ and $M_1$ use Yao's 
            garbled circuit protocol~\cite{yao1982protocols} to run $(pk, s)
            \gets \bkslprgen(1^{\lambda})$.
            Yao's protocol outputs ($pk, e_b = sh_b^{(s)}$) to $M_b$. Here, $s$ is a
            secret key for the LPR encryption scheme.

        \item The other step of setup is step~\ref{e:diffiehellman} from Figure~\ref{fig:beaver-and-bgi}.
        \end{myenumerate4}

\begin{center}
\vspace{-3mm}
\textit{Model-loading phase}
\vspace{-3mm}
\end{center}

        \begin{myenumerate4}
        \setcounter{enumi}{2}
        \item \label{e:generateB} $M_b$ samples 
            $sh_b^{(\mymatrix{B})}
            \in_R \mathbb{Z}_p^{n \times m}$. 

        \item \label{e:converttoshares} $M_0$ and $M_1$ use Yao's 
            protocol to convert shares of each column of
                  $\mymatrix{B}$, that is, 
                  $sh_b^{(\mymatrix{B}[i])}$ for $i \in \{1, \ldots, m\}$, to 
                  $sh_b^{(B[i] \cdot s)}$, where $B[i] \in R_p$ is the polynomial
                    encoding of the column vector $\myvec{B}[i]$. The polynomial
                    encoding is standard and based on an application of Chinese remainder
                    theorem (CRT) to ring $R_p$~\cite{boyle2019homomorphic}.

        \item Other steps of model-loading are steps~\ref{e:receivemodelparams}
            and~\ref{e:maskmodelparams} from Figure~\ref{fig:beaver-and-bgi}.
        \end{myenumerate4}

\begin{center}
\vspace{-3mm}
\textit{Preprocessing phase}
\vspace{-3mm}
\end{center}

        \begin{myenumerate4}
        \setcounter{enumi}{5}
        \item \label{e:triplegeneration} (Generate Beaver triple shares) $M_b$ does the following.
        \begin{myenumerate3}
        \item Samples 
                  $sh_b^{(\myvec{a})} \in_R \mathbb{Z}_p^{1 \times n}$ and converts it to its polynomial form
                  $sh_b^{(a)}$. 

        \item \label{e:converttoct} (Encrypts $a$) Sends $C^{sh_b^{(a)}} 
                    \gets  \lprenc(pk, sh_b^{(a)})$ to $M_{1-b}$, 
                    receives $C^{sh_{1-b}^{(a)}}$ from  $M_{1-b}$, 
            and computes $C^{a} = C^{sh_b^{(a)}} + C^{sh_{1-b}^{(a)}}$ using the
            additively homomorphic property of LPR.

       \item (Multiplies $a$ with $B[i]$) 
        \label{e:hssmult}
        For each $i \in \{1, \ldots, m\}$,
        computes $sh_b^{(B[i]  \cdot a)} =
        \textrm{Mult}(sh_b^{(B[i] \cdot s)}, C^{a})$ 
            using the HSS
        multiplication instruction. $M_b$ then converts $sh_b^{(B[i]  \cdot a)}$ to
            its vector form $sh_b^{(\myvec{B}[i]  \odot \myvec{a})}$, 
                    where
$\odot$ denotes component-wise multiplication. $M_b$ 
        computes $sh_b^{(\myvec{c}[i])} = \sum_{j=1}^{j=n} sh_b^{(\myvec{B}[i]
\odot \myvec{a})} [j]$.
        \end{myenumerate3}

        \item Finally, perform Step~\ref{e:fsskeygeneration}
            from Figure~\ref{fig:beaver-and-bgi}.

        \end{myenumerate4}

\begin{center}
\vspace{-3mm}
\textit{Online phase} is as in Figure~\ref{fig:beaver-and-bgi}
\vspace{-3mm}
\end{center}
}
\hrule
\caption{\Sys's protocol after its reducing-TEE-code design step. This protocol does not show
the packing optimization, which is illustrated separately in
Figure~\ref{f:packing}.}
\label{fig:protocol-after-outsourcing}
\label{f:protocol-after-outsourcing}
\end{figure*}

\subsection{Promise and perils of \hssbkslprname HSS}
\label{s:hssissues}

A key property of \hssbkslprname is that it allows a client to outsource
computation to two servers that do not interact with each other. 
However, as described, \hssbkslprname is
not suitable for Beaver triple generation, for two reasons. 
First, \hssbkslprname requires three parties where one of them supplies \hssbkslprname keys and encodings (encryptions and shares)
of program inputs. However, in \sys's setup, there are only two
parties---machines $M_0$ and $M_1$. They have shares of a vector $\myvec{a}$ and a matrix $\mymatrix{B}$, and
require shares of $\myvec{c} = \myvec{a} \cdot \mymatrix{B}$. Therefore, how should $M_0, M_1$ obtain 
(i) \hssbkslprname keys, (ii) 
        ciphertexts for input $\myvec{a}$, and (iii) 
            shares of $\mymatrix{B} \cdot s$?

Second, the dimension $N$ of the input space $R_p$ is 
large, for example, $2^{12}$ or $2^{13}$, to ensure the security of LPR ciphertexts~\cite{HomomorphicEncryptionSecurityStandard, lyubashevsky2010ideal}. But
oftentimes the vector length in ML models, denoted by $n$, 
is smaller than $N$. For instance, a CNN for the MNIST
dataset~\cite{lecun2010mnist} has vectors with 128
entries~\cite{liu2017oblivious}. When $n < N$, a mapping of vectors or matrix rows
(of a Beaver triple) to
degree $N-1$ polynomials in $R_p$ wastes space and incurs unnecessary \cpu and
network overhead (relative to a mapping that would not waste space in the
polynomial).\footnote{When $n > N$, the vector-matrix product is
split into smaller products akin to block-matrix multiplication. In this case,
the last product has a vector of size $n - \floor{n/N} \cdot N$, which is 
$\leq N$.}

\subsection{\Sys's protocol that incorporates \hssbkslprname HSS}
\label{s:hssfixes}

Figure~\ref{f:protocol-after-outsourcing} shows \sys's protocol for the second
part of its reducing-TEE-code step. This protocol incorporates \hssbkslprname to
generate Beaver triple shares, while addressing the aforementioned issues, as follows.

First, the protocol
adapts \hssbkslprname for two parties
by using the general-purpose Yao's garbled circuit
protocol~\cite{yao1982protocols} to simulate the client's role. Yao's
        protocol generates \hssbkslprname keys (step~\ref{e:hsskeygen} under
setup in Figure~\ref{f:protocol-after-outsourcing}) and shares of $\mymatrix{B} \cdot s$ to supply to the \hssbkslprname
    Mult instruction
(step~\ref{e:converttoshares} under model-loading in
Figure~\ref{f:protocol-after-outsourcing}).

Second, the protocol allows $M_0$ and $M_1$ to generate LPR ciphertexts for
\myvec{a} using the additively homomorphic properties of LPR.
In particular, $M_0$ and $M_1$ generate ciphertexts for shares of $\myvec{a}$,
exchange them, and add them to get a ciphertext for $\myvec{a}$
(step~\ref{e:converttoct} in Figure~\ref{f:protocol-after-outsourcing}). With this change and the one
above, the first challenge of applying \hssbkslprname to Beaver triple generation 
is addressed. 

\begin{figure}[t!]
\vspace{-1mm}
\centerline{\includegraphics[width=3.35in]{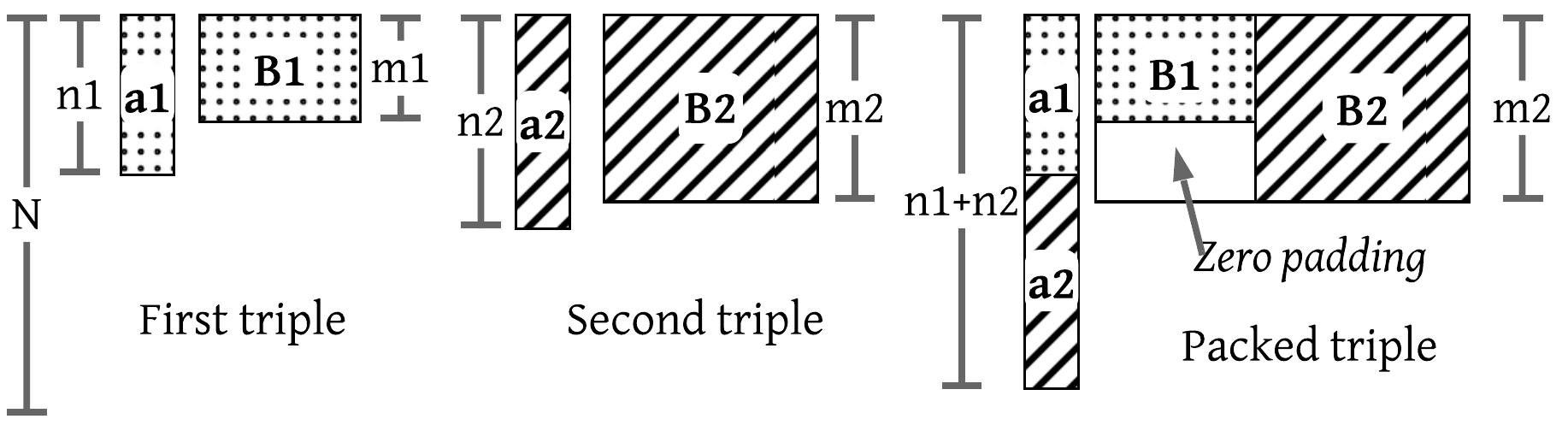}}
\caption{Packing scheme for triple generation.}
\label{f:packing}
\label{fig:packing}
\end{figure}

Third, the protocol 
    addresses the inefficiency caused by mapping small vectors in $\mathbb{Z}_p^n$ to
degree $N-1$ polynomials in $R_p$ by \emph{packing} 
        multiple smaller triples into an $N$-sized triple. Figure~\ref{fig:packing} depicts the overall idea. Say
that \sys needs to generate two triples 
                                    ($\myvec{a_1}, \mymatrix{B_1}, \myvec{c_1}$) 
and
                                    ($\myvec{a_2}, \mymatrix{B_2},
\myvec{c_2}$) for different layers of the same model, or different layers
across models, or different requests to the same layer of a model. Then, instead of running triple generation (step~\ref{e:triplegeneration} in
Figure~\ref{f:protocol-after-outsourcing}) separately for the two triples, 
    \sys runs a single instance of triple generation.

\emparagraph{Cost analysis.}
Relative to the protocol in Figure~\ref{f:beaver-and-bgi}, the cost of the setup and model-loading phases
increases because of the addition of Yao's protocol. However, Yao is used only
during setup and model-loading phases, and thus its cost gets amortized across
many inference requests, as setup runs once and model-loading runs once per
model.

The preprocessing phase adds inter-server network overhead to generate
ciphertexts for $\myvec{a}$ (step~\ref{e:converttoct} in
Figure~\ref{f:protocol-after-outsourcing}); this overhead is a small multiple of
$\myvec{a}$'s size due to the packing technique. The preprocessing phase 
adds \cpu cost, mainly due to the calls to LPR encryption function and the HSS 
Mult instruction.
(The online phase does not change relative to 
Figure~\ref{f:beaver-and-bgi}, so its costs do not get affected.)

\emparagraph{Security analysis.}
The protocol's security follows from the security of \hssbkslprname and Yao's
garbled circuits. In particular, the protocol satisfies the single-TEE
definition in \S\ref{s:security-definition}. Appendix~\ref{a:single-tee-proof} contains the
proof. 
\section{Details of distributing-trust step}
\label{s:fsskeygeneration}

The protocol has so far assumed a single TEE per server. 
    In particular, 
\begin{myenumerate4}[(i)]
\item step~\ref{e:diffiehellman} 
    in Figure~\ref{f:beaver-and-bgi} 
    uses the TEE $T_b$ at server $S_b$ to 
    set up a common seed for a PRF so that the TEEs
    at the two servers
    generate the same sequence of random values,

\item step~\ref{e:fsskeygeneration} in Figure~\ref{f:beaver-and-bgi} 
    uses the TEE machine $T_b$ at server $S_b$ to sample randomness $\myvec{r}$
    and output $sh_b^{\myvec{r}}$ to $M_b$, and

\item the same step 
    uses the TEE machine $T_b$ to 
    run $(k_0, k_1) \leftarrow 
        \fssgen(1^{\lambda}, \widehat{f}_{\myvec{r}}())$
    and output FSS key $k_b$ to $M_b$.
\end{myenumerate4}
    In this section, we
remove the single TEE limitation, by distributing
the computation in these steps over multiple, heterogeneous TEE machines. 

First off, in the multiple-TEE setting, 
both servers $S_0, S_1$ consist of a group of three TEEs denoted by 
$T_0^{(0)},
T_0^{(1)}, T_0^{(2)}$ 
and $T_1^{(0)}, T_1^{(1)}, T_1^{(2)}$ respectively. 

Then, to distribute the first part above (under bullet (i)), 
each pair of TEEs ($T_{0}^{(i)}, T_{1}^{(i)}$) for $i \in \{0, 1, 2\}$ 
establishes a common PRF seed, say $seed_i$, using the Diffie-Hellman key exchange
protocol~\cite{diffie1976new, chevalier2009optimal}. 
A common PRF seed ensures that both TEE machines in a pair (where one comes from
either server) generate the same sequence of random values. 

Next, to distribute the second part above (generation of $\myvec{r}$ under
bullet (ii) above), 
each TEE samples randomness locally and considers it to be its share of $\myvec{r}$. In more detail, 
let $r$ be a component of the randomness
vector $\myvec{r}$, and $r_0, r_1, r_2$ be uniformly
random elements in $\mathbb{Z}_p$ such that $r_0 + r_1 +
r_2 = r \pmod{p}$. 
Then, each TEE $T_b^{(i)}$ for $i \in \{0, 1, 2\}$ executes 
        the procedure GenRand in Figure~\ref{fig:beaver-and-bgi} to sample $r_i$. 
    Further, TEE $T_b^{(i)}$ 
        sends a share of $r_i$, that is, $sh_b^{r_i}$, to $M_b$.
\begin{figure}[t]
{\begin{algorithmic}
    \Function{\textbf{GenRand}}~{($b, seed_i$):}
        \State $r_i \gets PRF_{seed_i}(counter) \pmod{p}$
        \State $sh_0^{r_i} \gets PRF_{seed_i}(counter + 1) \pmod{p}$
        \State $sh_1^{r_i} \gets r_i - sh_0^{{r}_i} \pmod{p}$
        \State \Return $sh_{b}^{r_i}$
    \EndFunction
\end{algorithmic}
}
\caption{Procedure that TEE $T_b^{(i)}$ runs to distributively generate randomness
needed for the BGI protocol ($\myvec{r}$ in step~\ref{e:fsskeygeneration} in
Figure~\ref{fig:beaver-and-bgi}).}
\label{f:genrand}
\label{fig:genrand}
\end{figure}
Machine $M_b$ receives shares 
$sh_b^{r_0},
sh_b^{r_1}$, and $sh_b^{r_2}$ 
from $T_b^{(0)}$, $T_b^{(1)}$, and
$T_b^{(2)}$ respectively, and computes 
$sh_b^{r} = 
sh_b^{r_0} + 
sh_b^{r_1} + sh_b^{r_2}$.

Finally, to distribute the third part (FSS key generation under bullet (iii)
above), 
a natural starting point is to use   
a general-purpose three-party MPC
    protocol~\cite{yao1982protocols, goldreich1987play}. 
However, general-purpose protocols are expensive. 
Instead, \sys observes that the task at hand is to compute FSS keys for only the
    point and interval functions as they
can express all common non-linear functions (\S\ref{s:beaver-and-bgi}). 
Thus, \sys uses a customized 
    protocol for the point
        and interval functions, 
    thereby reducing overhead in comparison to a general solution.
In the rest of this section, we focus on this protocol; we first give a brief
overview of $\fssgen$ and then describe the protocol.

\emparagraph{Overview of \fssgen.}
\label{s:fssbackground}
\label{s:fss}
The FSS scheme of Boyle et al. is sophisticated~\cite{boyle2015function,
boyle2016function}.
        Moreover, 
            one does not need to understand its low-level details 
            to understand \sys's protocol.
For these reasons, 
    we describe only the notable
    aspects of Boyle et al.'s scheme. 

\begin{figure}[t!]
\centerline{\includegraphics[width=3.35in]{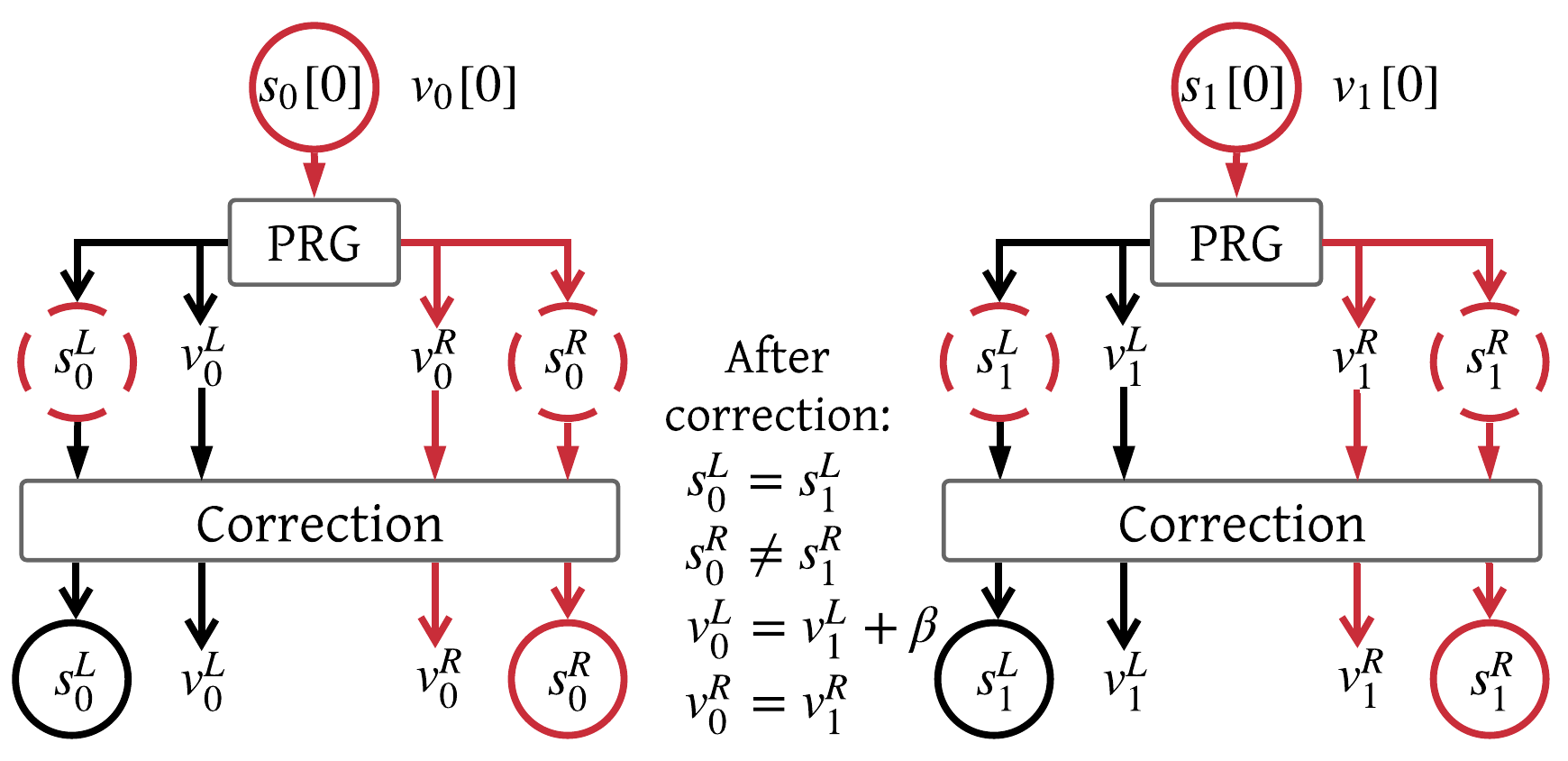}}
\caption{%
Pictorial depiction of the underlying computation in \fssgen~\cite{boyle2015function,
boyle2016function}. The procedure expands two paths in a binary tree; for each
node in the path, it invokes a PRG.}
\label{f:fssgenpic}
\end{figure}

Figure~\ref{f:fssgenpic} shows the key idea 
    behind $\fssgen$. 
Essentially, the keys $k_0, k_1$ output by $\fssgen$
            are paths from
        root to leaf nodes in
            two correlated trees.
At each step of path traversal,
        $\fssgen$
    takes a random string and performs two computations: (i)~expands
        the random string to two random strings for the two
        children, and (ii)~corrects the value
            of the children's strings so that they
        satisfy a certain constraint.
These two computations are depicted in Figure~\ref{f:fssgenpic}
    as the ``PRG'' (pseudorandom number generator) and 
        ``correction'' blocks.
The former is usually instantiated using AES and incurs significant expense when 
    performed inside MPC (we use ABY\textsuperscript{3}~\cite{mohassel2018aby}
in our implementation; \S\ref{s:impl}). Meanwhile, the
correction step is cheap as it mainly consists of XORs, which are
typically efficient in MPC.
Therefore, \sys's goal 
    is to reduce the cost of invoking the PRG.

\emparagraph{\Sys's protocol for \fssgen.} 
\Sys's idea is to bring the invocation of the PRG ``outside'' of 
        the general-purpose MPC protocol. 
Suppose that when the three TEE machines at server $S_b$, that is,
$T_b^{(0)},
T_b^{(1)}, T_b^{(2)}$,
    reach the PRG step, they hold
    XOR-shares of a string $x \in \{0,1\}^{\lambda}$ that they have to expand
    by applying a PRG, $G$. 
    That is, $T_b^{(j)}$
        holds
        $x_j$ such that $x_0 \oplus x_1 \oplus x_2 = x$. Then, to obtain strings that are computationally indistinguishable from the
shares of $G(x)$, $T_b^{(j)}$ does the following:
\begin{myitemize2}
    
    \item Splits $x_j$ into three blocks $x_j = x_j[0] \| x_j[1] \| x_j[2]$, where
    $\|$ denotes string concatenation. 

    \item Sends $x_j[i]$ and $x_j[k]$ (for $i, k \in \{0, 1, 2\}$ such that $i
\neq j, k \neq j$) to $T_b^{(i)}, T_b^{(k)}$, respectively. 
        After this step, 
        $T_b^{(0)}$, 
$T_b^{(1)}$, and
$T_b^{(2)}$, 
                obtain $x[0]$, $x[1]$, and $x[2]$,
            respectively. That is, the TEE machines obtain blocks of $x$
            from their shares of $x$.

    \item Invokes a PRG locally, say $g$, over $x[j]$ 
        to expand it to the same length as the output of $G(x)$. 

    \item Treats the output of $g$ as its XOR-share of
            $G(x)$, and continues onto
                the next step in \fssgen.
\end{myitemize2}

\emparagraph{Security and cost analysis.}
    Appendix~\ref{a:prf-proof} 
    proves that the output of $g(x[j])$ is indistinguishable 
        from a share of $G(x)$ even if 
        two blocks of the seed $x$ of $G$
    are revealed to a distinguisher.
    Further,
        Appendix~\ref{a:multiple-tee-proof} 
        proves that 
        \sys's protocol with multiple TEEs
    meets the multiple-TEE security definition in
\S\ref{s:security-definitions}. Meanwhile, the benefit of the PRG optimization
    is a reduction in both \cpu and network
    overhead relative to a general MPC solution,
        as each TEE invokes a PRG natively on its \cpu, rather than inside
            the MPC framework.
For instance, the network transfers between the TEE machines reduce from 1.6~MB in
ABY\textsuperscript{3}~\cite{mohassel2018aby} to 60~KB 
with the optimization (\S\ref{s:eval:microbenchmarks}).
\section{Implementation}
\label{s:impl}

We have implemented a prototype of \sys
(\S\ref{s:arch}, \S\ref{s:beaver-and-bgi}--\S\ref{s:fsskeygeneration}).
Our prototype builds on existing libraries. 
It implements the FSS primitives and the BGI protocol 
(\S\ref{s:beaver-and-bgi}) using the 
    libFSS library~\cite{libfss}.  
It implements the \hssbkslprname HSS scheme and our extensions 
        to the scheme
        (\S\ref{s:triplegeneration})
    on top of
    Microsoft's SEAL library~\cite{sealcrypto}. We borrow small pieces of code from ABY\textsuperscript{3}~\cite{mohassel2018aby} and 
    OpenSSL to implement the 
        secure computation protocol 
    for \fssgen (\S\ref{s:fsskeygeneration})
        atop the Asylo framework~\cite{asylo} for Intel SGX. Finally, \sys's various components (\S\ref{s:arch}) communicate
    over the gRPC RPC framework~\cite{gRPC}. 
In total, \sys's prototype adds 17,000 lines of C++ on top of 
    existing libraries; this number is measured using the sloccount Linux
utility~\cite{sloccount}.
\section{Evaluation}
\label{s:eval}
\label{s:evaluation}

Our evaluation answers the following questions:

\begin{myenumerate4}
    \item What are \sys's overheads for computing vector-matrix products, ML non-linear
    functions, and performing inference over popular ML models? 

    \item How do \sys's overheads compare
    to those of the state-of-the-art cryptography-based works?

    \item How accurately can \sys perform inference?

    \item How big is \sys's software TCB and how does its size compare to the TCB of systems that run ML inference completely inside TEEs?
\end{myenumerate4}

\begin{figure}[t]
\footnotesize
\centering

\begin{tabular}{
@{}
*{1}{>{\raggedright\arraybackslash}b{.027\textwidth}}  @{ }
*{1}{>{\raggedleft\arraybackslash}b{.1\textwidth}}  @{ }
*{1}{>{\raggedleft\arraybackslash}b{.055\textwidth}}  @{ }
*{1}{>{\raggedleft\arraybackslash}b{.055\textwidth}} @{ }
*{1}{>{\raggedleft\arraybackslash}b{.06\textwidth}}  @{ }
*{1}{>{\raggedleft\arraybackslash}b{.09\textwidth}}  @{ }
*{1}{>{\raggedleft\arraybackslash}b{.045\textwidth}} @{ }
@{}
}

&  & & RAM & network & & \\
vendor  &   type & vCPUs  &  (GB) & (Gbps) &  processor & loc. \\

\midrule
 AWS    & m5.4xlarge      &   16   & 64       & 10        & Xeon & CA \\
 Azure  & D16s-v3          &   16   & 64       & 8         & Xeon & CA \\
 Azure  & L8s-v2        &   8   & 64       & 3.2        & AMD EPYC  & WA \\
 Azure  & DC1s-v2       &   1   & 4       & 2         & Xeon-SGX   & VA \\
\bottomrule
\end{tabular}
\caption{Machines used in our experiments.}
\label{fig:testbed}
\label{f:testbed}
\end{figure}
 
\begin{figure}[t]
    \footnotesize
    \centering
    
    \begin{tabular}{
        @{}
        *{1}{>{\raggedright\arraybackslash}b{.1\textwidth}}  @{ }
        *{1}{>{\raggedleft\arraybackslash}b{.04\textwidth}}
        *{1}{>{\raggedleft\arraybackslash}b{.04\textwidth}}
        *{1}{>{\raggedleft\arraybackslash}b{.04\textwidth}}
        *{1}{>{\raggedleft\arraybackslash}b{.04\textwidth}}
        *{1}{>{\raggedleft\arraybackslash}b{.04\textwidth}}
        *{1}{>{\raggedleft\arraybackslash}b{.04\textwidth}}
        *{1}{>{\raggedleft\arraybackslash}b{.04\textwidth}}
        *{1}{>{\raggedleft\arraybackslash}b{.04\textwidth}}
        @{}
        }
        & \multicolumn{8}{c}{\cpu time} \\
        \midrule

        & 
        \multicolumn{2}{r}{VecToPoly} & 
        \multicolumn{2}{r}{PolyToVec} & 
        \multicolumn{2}{r}{LPR.Enc} & 
        \multicolumn{2}{r@{}}{HSS.Mult} \\

        \multicolumn{1}{@{}l}{\textbf{for HSS}} &
        \multicolumn{2}{r}{ 185.0 $\mu$s} &
        \multicolumn{2}{r}{ 168.4 $\mu$s} & 
        \multicolumn{2}{r}{ 4.9 ms} &
        \multicolumn{2}{r@{}}{ 3.6 ms} \\

        \midrule
        & \multicolumn{5}{c}{\cpu time} & \multicolumn{3}{r@{}}{network transfers}\\
        \hline

        \multicolumn{1}{@{}l}{\textbf{for FSS}} & 
        \multicolumn{2}{r}{\textit{Single TEE}} & 
        \multicolumn{3}{r}{\textit{Multiple TEEs}} & 
        \multicolumn{3}{r@{}}{\textit{Multiple TEEs}} \\

        \multicolumn{1}{@{}l@{}}{Gen (pt. fn.)} &
        \multicolumn{2}{r}{ 47.2 $\mu$s} & 
        \multicolumn{3}{r}{0.33 ms} & 
        \multicolumn{3}{r@{}}{ 43.8 KB} \\

        \multicolumn{1}{@{}l@{}}{Eval (pt. fn.)} &
        \multicolumn{2}{r}{ 22.1 $\mu$s} & 
        \multicolumn{3}{r}{ 147.5 $\mu$s} & 
        \multicolumn{3}{r@{}}{N/A} \\

        \multicolumn{1}{@{}l@{}}{Gen (int. fn.)} &
        \multicolumn{2}{r}{ 63.7 $\mu$s} & 
        \multicolumn{3}{r}{0.44 ms} & 
        \multicolumn{3}{r@{}}{ 55.7 KB} \\

        \multicolumn{1}{@{}l@{}}{Eval (int. fn.)} &
        \multicolumn{2}{r}{ 28.8 $\mu$s} & 
        \multicolumn{3}{r}{ 169.3 $\mu$s} & 
        \multicolumn{3}{r@{}}{N/A} \\

        \bottomrule
    \end{tabular}
    \normalfont\selectfont
    \caption{\cpu times and network transfers for HSS, FSS procedures, averaged over 1000 runs. Standard deviations (not shown) are
        within $1\%$ percent of the means. Network transfers are intra-domain.}
    \label{fig:microbenchmarks}
    \label{f:microbenchmarks}
\end{figure}

\noindent A summary of our evaluation results is as follows:
\begin{myitemize3}
    \item \Sys's \cpu and network overhead for computing a vector-matrix product is
    at least 1.5--2.2$\times$ and 1--60$\times$ lower depending on vector-matrix
    dimensions in comparison to state-of-the-art cryptography-based
    works (\S\ref{s:eval:overheadlinear}).

    \item \Sys's \cpu overhead is 11.6--45.1$\times$ higher for computing a non-linear function depending on the function
    in comparison to the popular Yao method in prior work.
        However, \sys's network overhead is 
    121.9--2819$\times$ lower
    (\S\ref{s:eval:overheadnonlinear}). 

    \item \Sys's \cpu overhead for private inference is higher than prior state-of-the work
        cryptography-based work by at most 14.2$\times$, while its inter-server network overhead is
        46.4--1448$\times$ lower depending on the ML model (\S\ref{s:eval:overheadinference}).

    \item Given that \cpu is cheaper than network resource, \sys's dollar cost for ML inference is 
        5.4--385$\times$ lower than prior work depending on the ML model
            (\S\ref{s:eval:overheadinference}).

    \item \Sys's inference accuracy is 1--2\% lower than TensorFlow's as it represents
            floating-point as fixed-point numbers (\S\ref{s:eval:accuracy}).

    \item \Sys runs 1,300 lines of code inside its TEEs, which is
        14.6--29.2$\times$ lower than the amount of code run inside TEEs by prior
        TEE-based works
    (\S\ref{s:eval:tcb}).
\end{myitemize3}

\paragraph{Method and setup.} 
We compare \sys's 
    two variants with single and multiple TEEs per server, which we call
    \sysstee and \sysmtee, to 
    several state-of-the-art baseline
systems.
For the performance-related questions, we compare \sys's variants to the
following cryptography-based systems.

\begin{figure*}[t]
\centering
\includegraphics[width=3.47in]{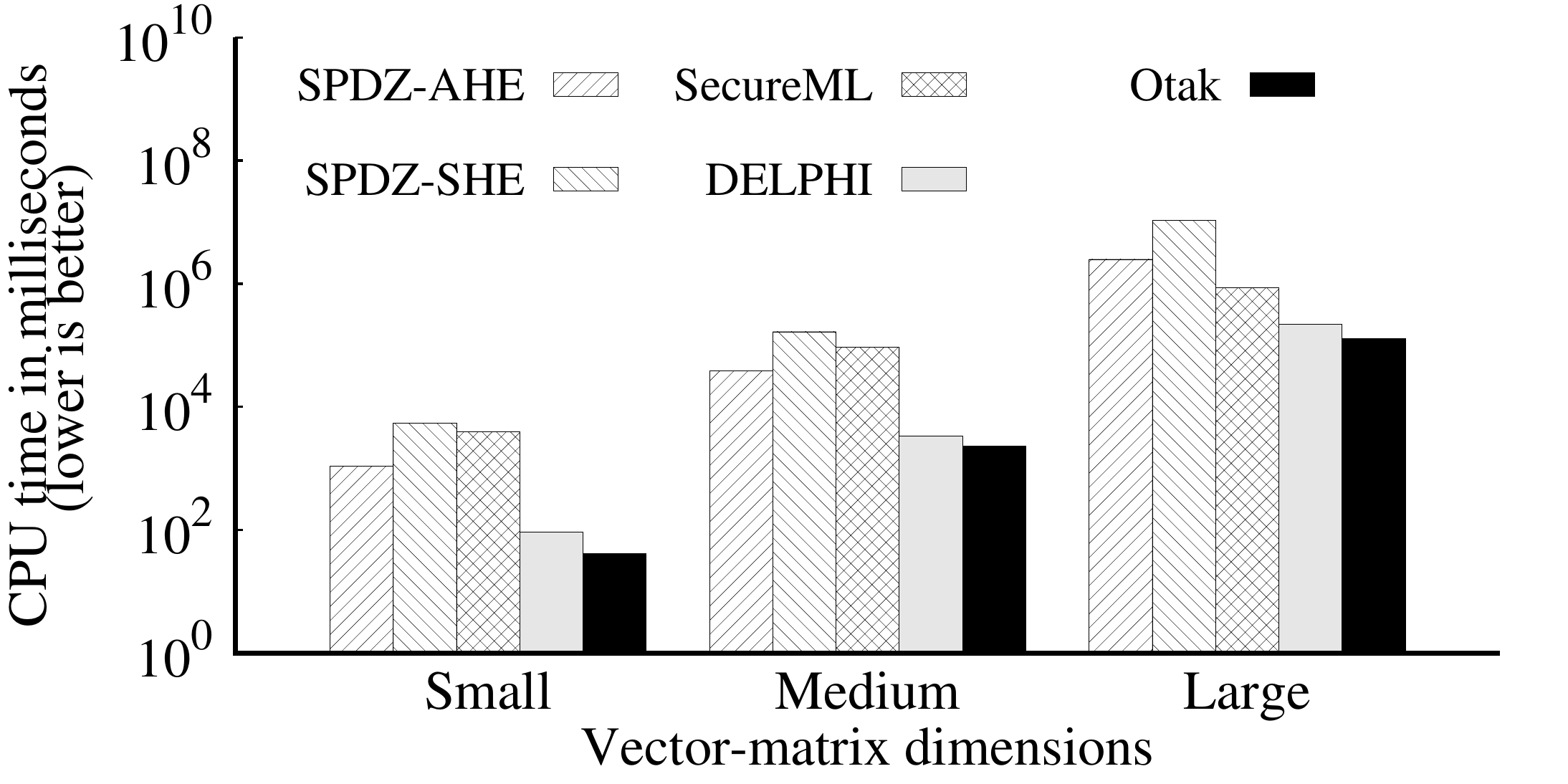}
\includegraphics[width=3.47in]{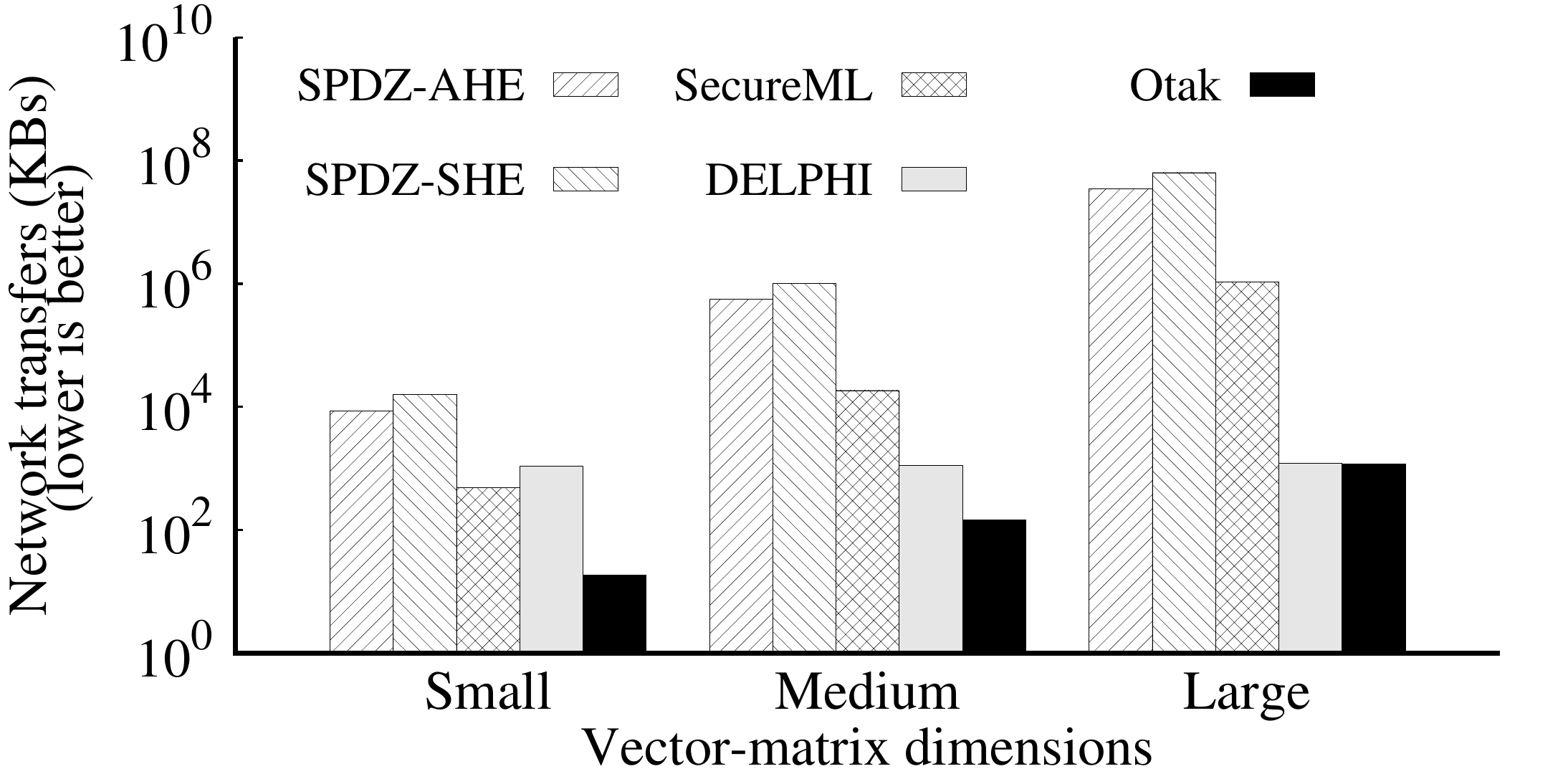}
\caption{\cpu times and network transfers for \sys and the baseline systems for privately computing vector-matrix
products.}
\label{f:overhead-vec-matrix}
\end{figure*}

\begin{myitemize}
    \item SecureML~\cite{mohassel2017secureml} and \spdz~\cite{keller2018overdrive,chen2019secure, damgaard2012multiparty} are 
    the state-of-the-art systems for the \sys-like setting where inference runs
    over two non-colluding servers that hold secret shares of model parameters
        and data points. We run the code
    of these systems while configuring them to provide
        honest-but-curious security.

    \item \delphi~\cite{mishra2020delphi} (which
        optimizes Gazelle~\cite{juvekar2018gazelle}) is a recent optimized
        2PC system that runs inference between a client (that owns data points) and a
            server (that owns model parameters). We run \delphi as a two-server
            system where one server owns model parameters and the other owns
            data points. Note that \sys's comparison to \delphi is not
            apples-to-apples as \sys further secret shares model parameters
            and data points between the two servers so that neither learns them. However, we include
            \delphi as it is a recent optimized 2PC system.
\end{myitemize}

\noindent For the TCB related questions, we compare \sys to two 
SGX-based systems: TF Trusted and \privado. TF Trusted is 
    optimized to run TensorFlow Lite inside an Intel SGX TEE~\cite{tftrusted}, 
    while \privado runs a reduced Torch ML framework inside
SGX~\cite{tople2018privado}. 
We pick these two systems from the set of SGX-based systems as their TCBs
are either reported or
    can be measured from their publicly available code.

We run a series of experiments to answer the evaluation questions above. Our
experiments deploy 
a system and vary parameters such as vector-matrix dimensions, the non-linear function
(ReLU, Maxpool, Sigmoid, Tanh, Argmax), and the ML model. For the latter,
we use four datasets (MNIST~\cite{lecun2010mnist}, CIFAR-10~\cite{kriz2014cifar},
CIFAR-100~\cite{kriz2014cifar}, and ISOLET~\cite{dua2017isolet}) and five ML model architectures
(two FNNs and three CNNs) including the 32-layer ResNet-32 CNN~\cite{he2016deep} (Appendix~\ref{a:models} gives more details). 
    These datasets and models perform a variety of inference tasks including
image classification and speech recognition. 
Our experiments measure inference accuracy for these
models as well as resource consumption: \cpu time using \texttt{std::clock()}, real time (latency) using
\texttt{std::chrono::high\_resolution\_clock}, and network transfers using Linux
kernel's \texttt{/proc/net/dev}. 

Our testbed (Figure~\ref{f:testbed}) is a set of machines on Amazon EC2 and
Microsoft Azure. Within these cloud providers, we use both general-purpose and TEE machines. Cloud providers currently
offer only Intel SGX-based TEEs~\cite{sgxavailability}; we use two such machines
and a regular AMD machine as \sysmtee's three TEEs. 
In a real deployment, \sys-MTEE would use three different TEEs.

\subsection{Microbenchmarks}
\label{s:eval:microbenchmarks}

We begin 
by
presenting \cpu and network transfers for primitive operations in \sys's 
cryptographic protocols (\S\ref{s:beaver-and-bgi}, \S\ref{s:triplegeneration},
\S\ref{s:fsskeygeneration}). Figure~\ref{f:microbenchmarks} shows these
microbenchmarks.

For computing vector-matrix products, \sys generates
    Beaver triple shares using the \hssbkslprname scheme
(\S\ref{s:hssfixes}). 
The first part of the microbenchmarks figure
        shows
the \cpu times for \hssbkslprname operations: converting a vector to its polynomial encoding
and back (used in step~\ref{e:converttoshares} and step~\ref{e:hssmult} in
Figure~\ref{f:protocol-after-outsourcing}), generating LPR ciphertexts
(step~\ref{e:converttoct} in Figure~\ref{f:protocol-after-outsourcing}), and
performing \hssbkslprname multiplication (step~\ref{e:hssmult} in
Figure~\ref{f:protocol-after-outsourcing}).
These microbenchmarks are 
    over an \texttt{m5.4xlarge} EC2
instance (Figure~\ref{f:testbed}). Note that we configure 
    \hssbkslprname for polynomials
with $N = 8{,}192$ coefficients, where each coefficient is up to $52$-bits.
These
parameters are chosen according to the homomorphic
encryption standard for a 128-bit security~\cite{HomomorphicEncryptionSecurityStandard}.

For computing non-linear functions such as ReLU, \sys relies on the FSS scheme
that supports point and interval functions (\S\ref{s:beaver-and-bgi},
\S\ref{s:fsskeygeneration}). The second part of the 
    microbenchmarks figure shows \cpu times and
network transfers (between TEE machines) for FSS procedures ($\fssgen, \fsseval$) 
    for the two functions. 

\begin{figure}[t]
\footnotesize
\centering

\begin{tabular}{
@{}
*{1}{>{\raggedright\arraybackslash}b{.12\textwidth}}  @{ }
*{1}{>{\raggedleft\arraybackslash}b{.07\textwidth}}
*{1}{>{\raggedleft\arraybackslash}b{.07\textwidth}}
*{1}{>{\raggedleft\arraybackslash}b{.075\textwidth}}
@{}
}
        & \cpu &  network (local) & network (wide-area) \\
        
    \midrule
    \textbf{ReLU} \\
    Yao &
                    0.45 ms &
                    0 &
                     8.3 KB
                    \\
    \sysstee &
                     1.3 ms &
                     13.9 KB &
                     18.0 B
                \\
    \sysmtee &
                     5.8 ms &
                     336.8 KB &
                    18.0 B 
                \\
    \midrule
    \textbf{Sigmoid} \\
    Yao &
                    1.3 ms &
                    0 &
                     46.8 KB
                    \\
    \sysstee &
                     3.4 ms &
                     37.7 KB &
                     17.0 B
                \\
    \sysmtee &
                     18.5 ms &
                     989.1 KB &
                     17.0 B 
                \\
    \midrule
    \textbf{Tanh} \\
    Yao &
                    1.44 ms &
                    0 &
                     48.6 KB
                    \\
    \sysstee &
                     3.3 ms &
                     37.8 KB &
                     18.0 B
                \\
    \sysmtee &
                     16.7 ms &
                     989.9 KB &
                     18.0 B 
                \\
    \midrule
    \textbf{MaxPool} \\
    Yao &
                    0.45 ms &
                    0 &
                     12.2 KB
                    \\
    \sysstee &
                     2.9 ms &
                     23.5 KB &
                     89.0 B
                \\
    \sysmtee &
                     20.3 ms &
                     1.2 MB &
                     89.0 B 
                \\
    \midrule
    \textbf{Argmax} \\
    Yao &
                    0.48 ms &
                    0 &
                     12.5 KB
                    \\
    \sysstee &
                     3.2 ms &
                     24.8 KB &
                     105.0 B
                \\
    \sysmtee &
                     21.6 ms &
                     1.3 MB &
                     105.0 B 
                \\
    \bottomrule
    
\end{tabular}
\normalfont\selectfont
\caption{\cpu times and network transfers for computing non-linear
functions using Yao and \sys. The overheads of MaxPool and Argmax
depend linearly on the number of input entries; 
here, we show overhead per entry.}
\label{fig:overheadnonlinearfunctions}
\label{f:overheadnonlinearfunctions}
\end{figure}
\begin{figure*}[t]
\centering
\includegraphics[width=3.47in]{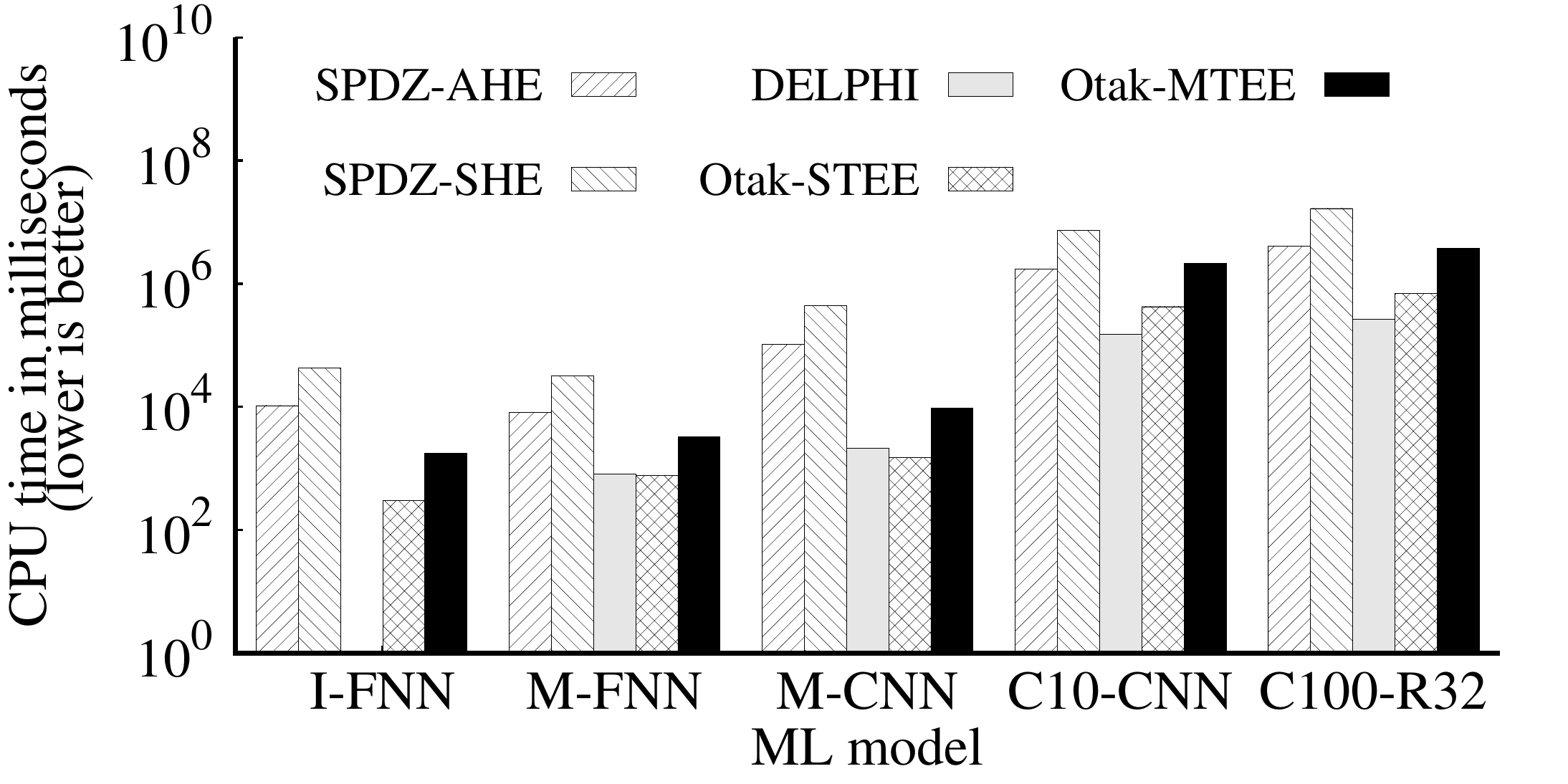}
\includegraphics[width=3.47in]{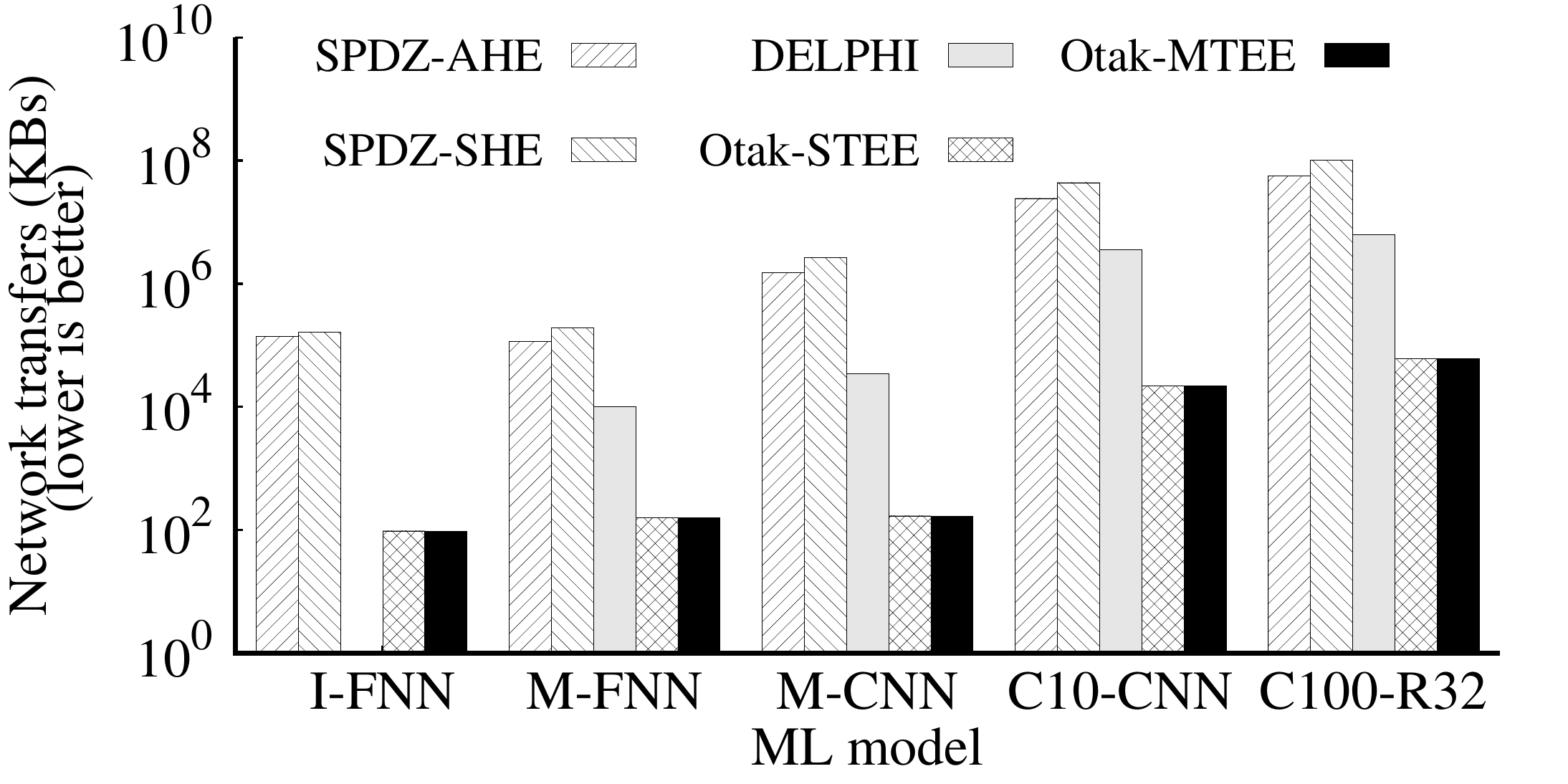}
\caption{\cpu time and (wide-area) network transfers for \sys and baseline systems for various ML models. \delphi does not support Tanh for
ISOLET-FNN (I-FNN). \sys's \cpu times are higher than prior work's
\cpu time. However, \sys reduces the expensive network transfers.
        SecureML (not shown in the figure) 
            currently only implements the ReLU function, and can, therefore, 
            encode just the MNIST-FNN (M-FNN) model. For this model, 
        its \cpu and network overhead is 9.8$\times$ and 46.4$\times$,
respectively, higher than \sys-MTEE's.}
\label{f:overhead-inference}
\end{figure*}

\subsection{Overheads of computing vector-matrix products}
\label{s:eval:overheadlinear}

Figure~\ref{f:overhead-vec-matrix} shows the \cpu times and (wide-area) network 
    transfers of various systems 
        for computing vector-matrix products
   while varying the dimensions of the vector and the matrix.
The ``small'', ``medium'', and ``large'' dimensions
    correspond to matrices with $128 \times 128$, $1024 \times 1024$,
$8192 \times 8192$ entries.

\emparagraph{\cpu overhead.}
For a particular matrix dimension, 
    \sys 
     (whose two variants do not differ in how they
                compute vector-matrix products),
        consumes a lower amount of \cpu than SecureML, \spdz (its both variants
            based on additive and somewhat homomorphic encryption), and
\delphi. For instance, \sys's \cpu consumption is 6.7--98.2$\times$
    lower than SecureML's, and  
            1.5--2.2$\times$ lower than \delphi's.
    SecureML consumes a high amount of \cpu because
        it uses the expensive number-theoretic Paillier
        additive homomorphic encryption scheme~\cite{paillier99public}.
    \spdz and \delphi 
        use modern additively homomorphic encryption schemes along-with
            packing techniques~\cite{juvekar2018gazelle}.
    Meanwhile, \sys improves over these works 
            by using HSS which enables even more efficient packing
                for Beaver triple generation (\S\ref{s:hssfixes}).

\emparagraph{Network overhead.}
Like for \cpu, \sys's network transfers
    are lower than those in prior works.
For instance, \sys's overhead is 
    1--60$\times$ lower than \delphi's.
    This is because \delphi does not optimally use
            the domain of its underlying encryption scheme; \sys,
        instead, optimizes the input domain of HSS instructions
(\S\ref{s:hssfixes}).
    Note than for ``large'' matrix dimensions when the size of the vectors
    equals the dimension of the underlying polynomial ring ($8{,}192$),
        \sys's overhead is same as \delphi's as packing
        does not take effect.
    However, for other dimensions, \sys's packing helps reduce
        overhead.

\subsection{Overheads of computing non-linear functions}
\label{s:eval:overheadnonlinear}

Figure~\ref{f:overheadnonlinearfunctions} shows the overheads of privately computing ReLU, Maxpool,
    Sigmoid, Tanh, and Argmax using \sys's FSS-based protocol
(\S\ref{s:beaver-and-bgi},
    \S\ref{s:fsskeygeneration}) and Yao's garbled circuits---the protocol
commonly used in prior two-party
    systems.

At a high level, Yao's protocol incurs lower \cpu consumption than \sys,
    especially when \sys uses multiple TEEs. The reason is that the \fsseval
    procedure is more expensive than the evaluation procedure of a Yao's garbled
circuit. Besides, the use of MPC to generate FSS keys adds \cpu expense for the
multiple TEE case. On the other hand, 
    Yao
    incurs high wide-area network transfers as it 
    exchanges a large Boolean circuit representing of the non-linear function
    between the two servers. 
In contrast, 
    \sys's wide-area network transfers are small (for example, by a factor of
461$\times$ for ReLU) due to the network-efficient online phase 
        of \sys's FSS-based BGI protocol 
(step~\ref{e:fsseval} in Figure~\ref{f:beaver-and-bgi}).
    \sys does incur intra-domain (local) transfers between TEE machines and
        between TEE machines and general-purpose machines
(step~\ref{e:fsskeygeneration} in Figure~\ref{f:beaver-and-bgi}). 
    However, local transfers are cheap. 
    Indeed, popular cloud providers do not charge for intra-domain network
        transfers within a geographical zone~\cite{googlenetworkpricing, azurenetworkpricing}.
    Overall, the reduction in expensive wide-area network overhead outweighs the
        increase in cheaper \cpu and local network transfers.

The \cpu and network costs for \sys follow from microbenchmarks (Figure~\ref{f:microbenchmarks}). For instance, ReLU
calls the interval function six times. According to
Figure~\ref{f:microbenchmarks}, six calls
to the interval function with multiple TEEs incurs 334.2~KB in local 
transfers,
which is roughly what Figure~\ref{f:overheadnonlinearfunctions} reports.

\subsection{Overheads of private inference}
\label{s:eval:overheadinference}

\emparagraph{\cpu and network overhead.}
    Figure~\ref{f:overhead-inference} shows \cpu and (wide-area) network use for the various systems
        and ML models.

    \Sys's \cpu time is higher than the \cpu time in prior work. The reason is
    that \sys requires more 
        \cpu for non-linear
    functions (\S\ref{s:eval:overheadnonlinear}),  
        especially when it uses multiple TEEs.
    However, \sys reduces network transfers (\S\ref{s:eval:overheadlinear},
    \S\ref{s:eval:overheadnonlinear}) significantly relative to prior work. 
For instance, for the ResNet-32 
    CNN model over the CIFAR-100 dataset (cluster
        labeled C100-R32
        in Figure~\ref{f:overhead-inference}),
        \sys's both variants incur 60~MB of wide-area network transfers whereas
            \delphi consumes 6~GB, \spdz with additive homomorphic encryption
(\spdz-AHE) consumes 53.3~GB, and \spdz with somewhat homomorphic encryption
                consumes 96.4~GB. 
        Note that SecureML (which we do not show in the figure) 
            currently only implements the ReLU function, so 
            it can encode just the MNIST-FNN (M-FNN) model. For this model, 
        its \cpu and network overhead is 9.8$\times$ and 46.4$\times$,
respectively, higher than \sys's.

\begin{figure}[t]
    \footnotesize
    \centering
    
    \begin{tabular}{
        @{}
        *{1}{>{\raggedleft\arraybackslash}b{.08\textwidth}}
        *{1}{>{\raggedleft\arraybackslash}b{.043\textwidth}}
        *{1}{>{\raggedleft\arraybackslash}b{.053\textwidth}}
        *{1}{>{\raggedleft\arraybackslash}b{.053\textwidth}}
        *{1}{>{\raggedleft\arraybackslash}b{.066\textwidth}}
        *{1}{>{\raggedleft\arraybackslash}b{.066\textwidth}}
        @{}
        }
        & I-FNN & M-FNN & M-CNN & C10-CNN & C100-R32 \\
        \midrule

        SecureML & - & \$0.48 & - & - & - \\

        \spdz-AHE 
        & \$12.15 & \$10.26 & \$130.82 
        & \$2129.13 & \$4918.67 \\

        \spdz-SHE 
        & \$11.56 & \$17.75 & \$242.38 
        & \$4044.48 & \$9337.42 \\

        \delphi
        & - & \$0.48 & \$1.67 
        & \$169.73 & \$301.71 \\
        
        \sys-STEE
        & \cent0.84 & \$0.02 & \$0.03 
        & \$5.63 & \$10.68 \\

        \sys-MTEE
        & \$0.03 & \$0.05 & \$0.15 
        & \$31.43 & \$55.94 \\
        \bottomrule
    \end{tabular}
    \normalfont\selectfont
    \caption{Dollar costs for 1000 private predictions. 
    \Sys's dollar costs are lower as it reduces
    wide-area network transfers substantially.}
    \label{fig:bills}
    \label{f:bills}
\end{figure}

\emparagraph{Dollar costs.}
\Sys's \cpu use is higher than prior work
    while network use is lower.
To compare the systems using a common metric, we convert their resource
use to a dollar amount.

Figure~\ref{f:bills} shows estimated dollar costs for private inference for the
various systems. 
To do the conversion from resource overhead to dollars, we use a pricing model derived
from the machine and bandwidth prices of Azure and AWS (Appendix~\ref{a:pricingmodel}). This pricing
model charges $\$0.015$--$\$0.079$ for one hour of \cpu time depending on
machine type (SGX versus non-SGX), $\$0.05$ for
one GB of outbound network traffic, and zero for local network transfers. 
The figure shows that \sys's dollar cost, depending on \sys's variant, is 9.6--24$\times$ lower than SecureML's,
87--4360$\times$ lower than \spdz's, and 5.4--55.6$\times$ lower than \delphi's.

\emparagraph{Inference latency.}
Figure~\ref{f:overhead-inference-latency} shows the latency of performing inference
for the various systems (SecureML is not depicted in the figure, and its latency
for M-FNN is 150.2~ms).
Overall, \sys takes less time than SecureML (by 2.45$\times$) and SPDZ (by
2.62--262$\times$), 
but longer than \delphi 
(by up to 10.2$\times$)
to perform inference.
The difference to \delphi is
fundamental---\delphi targets a non-outsourced setting where model parameters are in plaintext while
\sys hides the model parameters from the provider. Owing to the difference in
setting, \delphi's operations are cheaper than \sys's in its online phase.

\begin{figure}[t]
\centering
\includegraphics[width=3.47in]{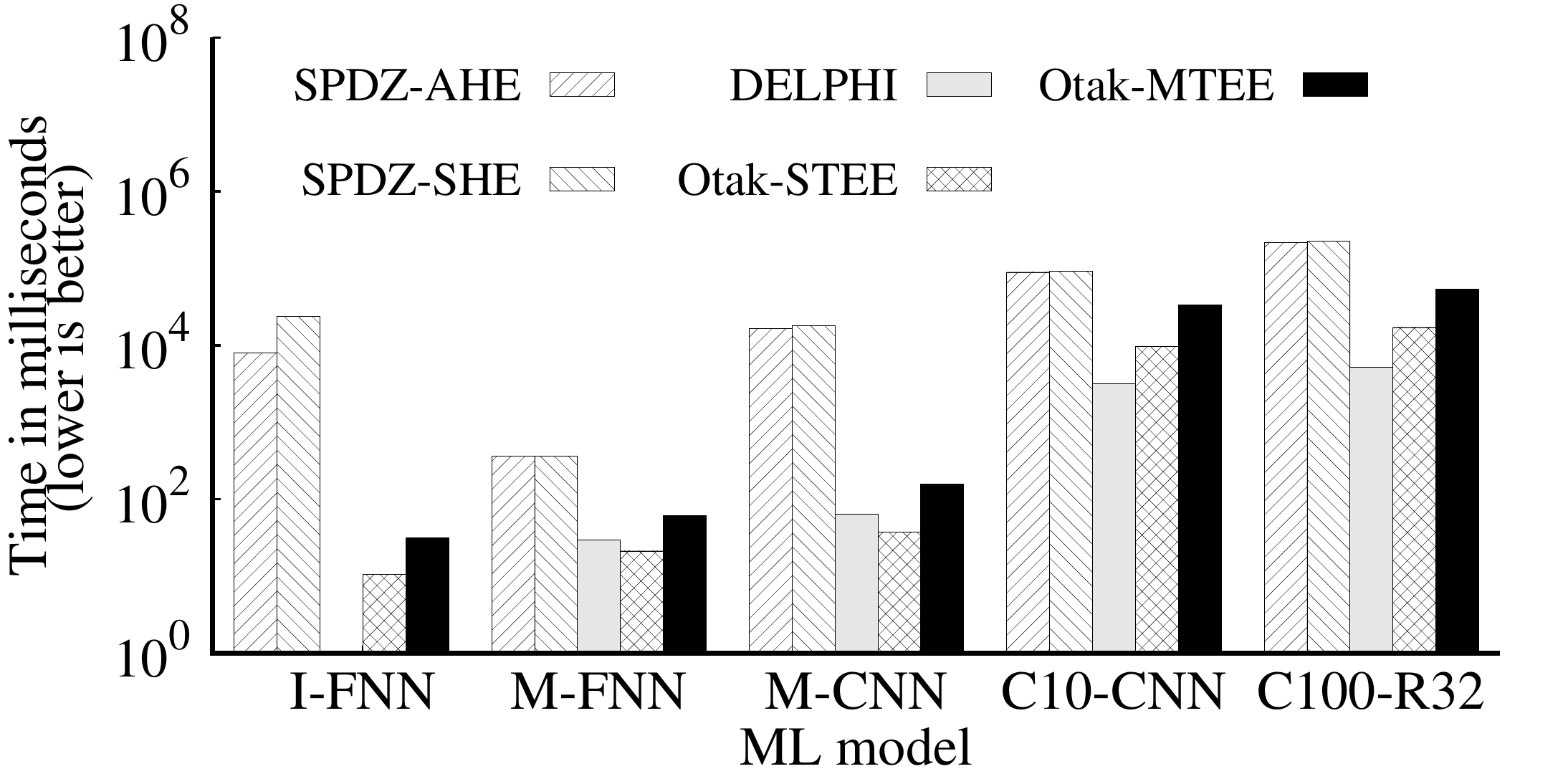}
\caption{Inference latency for \sys and baseline systems. \delphi does not
support Tanh for I-FNN. 
\Sys's difference to \delphi is 
fundamental---\delphi targets a non-outsourced setting where model parameters are in plaintext while
\sys hides the model parameters from the service provider. This difference in
setting enables \delphi to have a very efficient online
phase~\cite{mishra2020delphi}.}
\label{f:overhead-inference-latency}
\end{figure}

\emparagraph{Client-side costs for the data point owner}
\label{s:clientinferenceresults} are small. For
instance, for the largest model (C100-R32), the data point
owner expends 2~ms in \cpu time and 61.2~KB in network transfers to send secret shares of input image and receive the
output label.
            
\label{s:eval:accuracy}
\emparagraph{Accuracy.} We compare \sys's inference accuracy against
TensorFlow's for I-FNN, M-FNN, M-CNN, C10-CNN, and C100-R32. 
\Sys outputs the correct label 
    97.5\%, 
    97.6\%, 
    97.2\%, 
    81.6\%, 
    and
    68.3\%
    of the
time, while TensorFlow outputs correct label 
    99.0\%, 
    97.8\%, 
    98.3\%, 
    82.7\%, 
    and 
    69.8\% 
of times.  \Sys's accuracy is 1-2\% lower than TensorFlow's
because it approximates floating-point as fixed-point numbers while
TensorFlow does not require such approximation (\S\ref{s:beaver-and-bgi}).

\subsection{Trusted computing base (TCB)}
\label{s:eval:tcb}

\Sys improves over cryptography-based solutions by using TEEs. However, a downside is the trust on TEEs, which \sys mitigates by 
        keeping the size of the functionality inside the TEE small and distributing trust over
        heterogeneous TEEs. Here, we report the size of the
        functionality \sys runs inside its TEEs in terms of source lines of
code and compare it to the lines of code inside TEEs for prior TEE-based
systems.

\sys's code inside its TEEs is 1,300 lines. 
This TCB is 14.6$\times$ and 29.2$\times$ lower than the TCB sizes of
\privado and TF trusted, which have 19K and 38K source lines of code inside
the TEE~\cite{tople2018privado, tftrusted}. (For \privado we
report TCB size from their paper; for TF Trusted
\cite{tftrusted}, we count only the \emph{included} header files in the TEE code
as opposed to the entire codebase of library dependencies~\cite{gRPC,
tensorflowcoreapi}.)
\section{Conclusion}
\label{s:discussion}
\label{s:conclusion}

Outsourced private ML inference over two servers is an important problem and has
attracted considerable attention. Prior systems for this problem
    employ two-party secure computation (2PC) protocols that incur high overhead. This paper asked the question, can we accelerate
    2PC for this setting by employing trusted-hardware in a limited
capacity, and found the results to be encouraging. 
In particular, one can accelerate 2PC by one to two orders of
magnitude (\S\ref{s:eval}), by building on recent primitives such as function and homomorphic
secret sharing (\S\ref{s:beaver-and-bgi}--\S\ref{s:fsskeygeneration})---while restricting trusted-hardware 
    to a small computation 
    during preprocessing, and without trusting the hardware of a particular
vendor. By demonstrating these promising results, this paper opens up new avenues---not just for 
        two-server outsourced ML inference
            but also for private ML inference in other settings such as three
        servers.

\appendix
\section{MaxPool and Argmax}

As mentioned earlier (\S\ref{s:beaver-and-bgi}), efficient 
FSS constructions~\cite{boyle2015function, boyle2016function, boyle2019secure}
currently exist only
for two functions: a point function and an interval function.  These
functions can encode a spline function~\cite{boyle2019secure}, which in turn can encode ReLU,
and close approximations of Sigmoid and Tanh~\cite{amin1997piecewise}. However, 
a spline cannot directly encode the MaxPool and Argmax functions. Here,
we describe how \sys encodes MaxPool and Argmax using a composition of the point
and interval functions.

\subsection{Maxpool}
\label{a:max}
\label{a:maxpool}

Note that max over an array of inputs can be expressed
    using multiple instances of max over two inputs using a tournament-like 
    tree structure. So we focus on encoding the latter.

\emparagraph{Problem with a straightforward encoding.}
A straightforward approach to expressing the max of two elements $x, y \in
\mathbb{Z}$ is to use the $\textrm{sign}$ function: 
\[   
\textrm{sign}(x - y) = 
     \begin{cases}
       \text{1 (+ve),} &\quad\text{if $z \geq 0$;}\\
       \text{0 (-ve),} &\quad\text{otherwise.} \\ 
     \end{cases}
\]
If the sign of $x-y$ is 1, then the max equals $x$, else the max equals y. 
However, as noted earlier (\S\ref{s:beaver-and-bgi}), the reasoning above does
not work when operating over elements of $\mathbb{Z}_p$ (instead of elements
of $\mathbb{Z}$)
as $\mathbb{Z}_{p}$ is not an ordered ring (when
encoding both positive and negative integers). 

\emparagraph{\Sys's encoding of max.}
Observe that $\textrm{sign}(x-y)$ gives the right answer when both $x$ and
$y$ have the same sign, that is, when they are either both positive or both negative. 
For this case, we can write $\textrm{max}(x,y)=\textrm{ReLU}(x-y) + y$. 
    For the case when $x$ and $y$ have different signs, we can write
$\textrm{max}(x,y)=\textrm{ReLU}(x)+\textrm{ReLU}(y)$. 
Therefore, \sys expresses
max as
\[   
\textrm{max}(x,y) = 
     \begin{cases}
       \text{ReLU}(x-y) + y,  &\text{if } \textrm{sign}(x) =
\textrm{sign}(y)\text{;} \\ 
       \text{ReLU}(x)+\text{ReLU}(y), &\quad\text{otherwise.}\\
     \end{cases}
\]
To combine the two cases, define a function $b(x,y)$ as
\[   
\textrm{b}(x,y) = 
     \begin{cases}
       \text{1,} &\quad\text{if sign(x) = sign(y);} \\ 
       \text{0,} &\quad\text{otherwise.}\\
     \end{cases}
\]
Then, 
\begin{equation*}
\begin{aligned}
  \textrm{max}(x, y)  &= \textrm{b}(x,y) \cdot (\textrm{ReLU}(x-y) + y) \\
                      & + (1-\textrm{b}(x,y)) \cdot
(\textrm{ReLU}(x)+\textrm{ReLU}(y)).
\end{aligned}
\end{equation*}
Let $pf_{\alpha}^{\beta}(\cdot)$ denote the point function that outputs $\beta$ at the
point $\alpha$ and zero otherwise.
Then, given the value of $\textrm{ReLU}(x)+\textrm{ReLU}(y)$, one can see that
$pf_{0}^{1}(\textrm{ReLU}(x)+\textrm{ReLU}(y))$ equals 1 if both $x$ and $y$ are
negative and zero otherwise. Similarly, $pf_{0}^{1}(\textrm{ReLU}(x)+
\textrm{ReLU}(y)-x-y)$ equals 1 if both $x$ and $y$ are positive and zero otherwise.
Hence, $\textrm{b}(x, y)$ can be defined as
\begin{equation*}
  \begin{aligned}
    \textrm{b}(x,y) &= pf_{0}^{1}(\textrm{ReLU}(x)+\textrm{ReLU}(y)) \\
            &+ pf_{0}^{1}(\textrm{ReLU}(x)+ \textrm{ReLU}(y)-x-y).
  \end{aligned}
\end{equation*}
With the above definition of the function $b(\cdot, \cdot)$, \sys can express
$\textrm{max}$ using point and interval functions. 
Note that \sys evaluates max over two rounds. 
In the first round, it computes the shares of the inner parts of the functions
(that is, $\textrm{ReLU}(x-y) + y$ and
$\textrm{ReLU}(x)+\textrm{ReLU}(y)$), while in the second round, it computes the 
    outer parts of the functions.

\subsection{Argmax}
\label{a:argmax}

Like for Maxpool, \sys's goal is to express Argmax in terms of point and interval functions. 

Recall that the Argmax function outputs the index
    of the maximum entry in an array.
Let the inputs to argmax be $x_i$ for $i \in [1, n]$ and let
    $\mu$ equal the biggest value, 
        that is, $\mu = \textrm{max}(\{x_i\})$, computed using
MaxPool (Appendix~\ref{a:max}).
Then, one can express the output of argmax 
    as an array of $n$ entries, 
    with 0 at the $i$-th entry if $x_i \neq \mu$, and  
    the index $i$ at the $i$-th entry if $x_i = \mu$. 

Let $pf_{\alpha}^{\beta}(\cdot)$ denote the point function that outputs $\beta$
at the point $\alpha$ and zero otherwise.
Then, we define $i$-th entry of the Argmax's output as $pf_0^{i}(x_i - \mu)$,
where $\mu = \textrm{max}(\{x_i\})$.
\newcommand{\calU}{\mathcal{U}}
\section{Security of the distributing-trust step}

\begin{lemma}
\label{a:prf-proof}
\label{prf-proof}
    Let $G\colon X\to Y$ be a PRG. Define $G'\colon X \times X \times X \to Y$  as $G'(x_1,x_2,x_3)=G(x_1)\oplus G(x_2)\oplus G(x_3)$. Then, the following holds: for every $i,j \in \{1,2,3\}$ and $i \neq j$, 
    $$\left\{ \left(G'(x_1,x_2,x_3),\ x_i,x_j  \right) \right\}_{x_1,x_2,x_3 \xleftarrow{\$} X} \approx_c \left\{ \left(u,\ x_i,x_j \right) \right\}_{\substack{x_i,x_j \xleftarrow{\$} X\\ u \xleftarrow{\$} \calU}} $$
    ($\approx_c$ denotes computational indistinguishability).
   \par In other words, $G'$ is a PRG and moreover, is secure even if any two
blocks of the seed of the PRG is revealed to a distinguisher. 
    \end{lemma}
    
  \begin{proof}
  It suffices to prove the case when $i=2,j=3$ and the other cases follow
symmetrically. From the security of $G$, which says that an output of $G$ is
    indistinguishable from an element sampled uniformly at random, the following holds: 
  \begin{center} 
   $\left\{ \left(G(x_1) \oplus G(x_2) \oplus G(x_3),\ x_2,x_3  \right) \right\}_{x_1,x_2,x_3 \xleftarrow{\$} X}$ \\
    $ \approx_c \left\{ \left(u \oplus G(x_2) \oplus G(x_3),\ x_2,x_3 \right) \right\}_{\substack{x_2,x_3 \xleftarrow{\$} X\\ u \xleftarrow{\$} \calU}} $ \\
  \end{center}
   Since XOR-ing the uniform distribution with any fixed value still gives the same distribution, we have the following: 
    $$\left\{ \left(u \oplus G(x_2) \oplus G(x_3),\ x_2,x_3 \right) \right\}_{\substack{x_2,x_3 \xleftarrow{\$} X\\ u \xleftarrow{\$} \calU}}\equiv  \left\{ \left(u,\ x_2,x_3 \right) \right\}_{\substack{x_2,x_3 \xleftarrow{\$} X\\ u \xleftarrow{\$} \calU}} $$
    ($\equiv$ denotes perfect indistinguishability). 
    Combining the above two observations, we have the proof of the lemma.
  \end{proof}
  
  \begin{theorem} 
  \label{gen-theorem}
  Let $\func_{gen}: (1^{\lambda}, b, \widehat{f}_{\myvec{r}}) \mapsto k_b$ for $b \in \{0, 1\}$ denote the functionality (from \S\ref{s:fsskeygeneration}) that outputs BGI keys to $M_b$. For any PPT adversary $\adversary$ corrupting $M_b, T_b^{(i)}$ for $i \in \{0, 1, 2\}$, with access to TEEs $T_b^{(j)}, T_b^{(k)}$, for $j, k \in \{0, 1, 2\}$ and $j \neq i, k \neq i$, for every large enough security parameter $\lambda$, there exists a PPT simulator $\simr$ such that the following holds:
  
  For every $b \in \{0, 1\}$, randomness $r_b \in \{0, 1\}^{poly(\lambda)}$ $$\view_{\adversary}^{\func_b^{(j)}, \func_b^{(k)}}(1^{\lambda}, b; r) \approx_{c, \epsilon} Sim(1^{\lambda}, b, r, k_b)$$
  \end{theorem}
  
  \paragraph{Remark.} The proof essentially follows from Lemma~\ref{prf-proof} and the simulation security of the $\mathrm{ABY}^{3}$ protocol in the honest-majority setting.

\section{Security Proof}

\label{a:proof}
We first present the proof for the single-TEE case and later, we show how to extend this proof to the multiple-TEE setting. 

\subsection{Single-TEE}
\label{a:single-tee-proof}

Suppose there exists a PPT adversary $\adversary$ that compromises the general purpose machine $M_b$ and TEE machine $T_b$ for some $b \in \{0, 1\}$. We then construct a PPT simulator $\simr$ as follows. 

\begin{myenumerate}
    \item $\simr$ chooses a uniform random tape for $M_b$, $T_b$.

    \item \textbf{In setup phase:}
    \begin{myenumerate}
    	\item $\simr$ simulates the Diffie-Hellman protocol.
    
        \item $\simr$ runs the corresponding simulator of Yao's protocol (see \cite{lindell2008yao} for its description) such that the simulated circuit outputs the simulated share $(pk,e_b)$, of the HSS keys, to machine $M_b$.
    \end{myenumerate}
    
    \item \textbf{In model-loading phase:} for every layer in the model,
    \begin{myenumerate}
        \item $\simr$ runs the corresponding simulator of Yao's protocol \cite{lindell2008yao}; the simulated garbled circuit is programmed to output an element in $R_q^2$, chosen uniformly at random, to machine $M_b$.

        \item $\simr$ picks a matrix in $\mathbb{Z}_p^{m \times n}$, uniformly at random, for layer parameters, and sends it to machine $M_b$. 
    \end{myenumerate}

    \item \textbf{In preprocessing phase:} for every layer in the model,
    
    \begin{myenumerate}
        \item Simulator $\simr$ gets the public key $pk$ from the simulated setup phase. For every $i \in \{1,\ldots,m\}$, it sends encryption of $0$ to $M_b$.
        
        \item $\simr$ computes the simulator of the BGI protocol, on input $(1^{\lambda},b)$, to obtain the simulated key $k_b$. The functionality inside $T_b$ will now output the key $k_b$.
        
    \end{myenumerate}

    \item \textbf{In online phase:}
    \begin{myenumerate}
        \item Simulator $\simr$ sends shares $sh_b^{\myvec{0}}$, chosen uniformly at random, to machine $M_b$.
        
        \item For every layer of the ML model, $\simr$ computes the simulator of the secure Beaver multiplication.
        
        \item At the end of every execution of the FSS protocol (one per non-linear layer), the simulator sends a share, chosen uniformly at random,  to $M_b$.  
    \end{myenumerate}
\end{myenumerate}

\noindent We show that the real world distributions is computationally
indistinguishable to the simulated distributions via the standard hybrid argument.

\begin{myenumerate}
    \item $Hyb_0:$ This corresponds to the real world distribution where the model-owner, the datapoint-owner, secure hardware machine $T_b$, and general-purpose machine $M_{1-b}$ execute the system as mentioned in the description of the protocol.

    \item $Hyb_1$: In this hybrid, the simulator for the Diffie-Hellman key exchange is executed. $Hyb_0 \approx_c Hyb_1$ follows from the simulation security of the DH key exchange. 

    \item $Hyb_2:$ In this hybrid, we call the simulator of Yao's protocol on input $1^ {\lambda}$ and the circuit $Cir_{setup}$ that outputs honestly generated BKS-LPR keys: $pk, e_b \in_{R} R_q^2$. $Hyb_1 \approx_{c} Hyb_2$ follows from the simulation security of Yao's protocol.

    \item $Hyb_3:$ In this hybrid, we modify $Cir_{setup}$ such that it outputs simulated BKS-LPR keys: $pk, e_b \in_{R} R_q^2$ to machine $M_b$. $Hyb_2 \approx_{c} Hyb_3$ follows from the fact that honestly generated BKS-LPR key is indistinguishable from the simulated key.

    \item $Hyb_4$: In this hybrid, we again call the simulator of Yao's protocol $Sim_{Yao}$ on input $sh_b^{\mymatrix{B}[i]}$ and the circuit $Cir_{transform}$ that outputs additive secret shares $sh_b^{B[i] \cdot \myvec{s}} \in R_q^2$ for each $i$-th column $\in \{1,...,m\}$ to machine $M_b$. $Hyb_3 \approx_{c} Hyb_4$ due to the simulation security of Yao's protocol.

    \item $Hyb_5$: In this hybrid, we modify $Cir_{transform}$ such that it outputs a value in $R_q^2$, chosen uniformly at random, for each $i$-th column $\in \{1,...,m\}$ to machine $M_b$. $Hyb_4$ is identical to $Hyb_5$ due to the perfect security of the secret shares.

    \item $Hyb_6$: In this hybrid, we change the inputs sent by the model-owner and the machine $M_{1-b}$ to the general-purpose machine $M_b$. Instead of sending secret shares of the model parameters $\mymatrix{Y}$, model-owner sends a value in $ \mathbb{Z}_p^{n \times m}$, chosen uniformly at random. Similarly, instead of sending secret shares $sh_{b - 1}^{(\mymatrix{Y - B})}$, machine $M_{b - 1}$ sends a value in $\mathbb{Z}_p^{m \times n}$, chosen uniformly at random. $Hyb_5$ is identical to $Hyb_6$ due to the perfect security of additive secret sharing scheme.

    \item $Hyb_{7.(j)}$ for $j^{th}$ layer: In this hybrid, we change the ciphertext sent by the machine $M_{b - 1}$ to machine $M_b$ in the execution of the $j^{th}$ layer. Specifically, instead of sending encryption of $sh_{b - 1}^{\myvec{a}}$, machine $M_{b - 1}$ sends encryption of $\myvec{0} \in \mathbb{Z}_p^{1 \times n}$. 
    \par The following holds from the semantic security of the LPR scheme: (i) for $j \in \{1,\ldots,L-1\}$, where $L$ is the number of layers, $Hyb_{7.(j)} \approx_{c} Hyb_{7.(j+1)}$ and, (iii) $Hyb_{6} \approx_c Hyb_{7.(1)}$.
    
    \item $Hyb_8$: In this hybrid, we invoke the simulator of BGI protocol to obtained simulated key $k_b$ that is then sent to $M_b$. $Hyb_8 \approx_c Hyb_{7.(L)}$ due to the security of the BGI protocol.

	\item $Hyb_9$: This corresponds to the output distribution of $\simr$.
	\par Hybrids $Hyb_8$ and $Hyb_9$ are identically distributed.
\end{myenumerate}

\subsection{Multiple-TEE}
\label{a:multiple-tee-proof}
\noindent We now focus on the setting when there are multiple TEEs.

Suppose there exists a PPT adversary $\adversary$ that compromises the general purpose machine $M_b$ and TEE machine $T_b^{(i)}$ for some $b \in \{0, 1\}, i \in \{0, 1, 2\}$. Let $j,k$ be such that $j \neq i$ and $j \neq k$. We then construct a PPT simulator $\simr$ as follows.

\begin{myenumerate}
    \item $\simr$ chooses a uniform random tape for $M_b$, $T_b^{(i)}$. $\simr$ simulates the other TEEs $T_{b}^{j},T_b^{k}$, where $i \neq j \wedge i \neq k$. 

    \item \textbf{In setup phase:} execute the setup phase of the simulator described in the single-TEE setting. 
    
    \item \textbf{In model-loading phase:} execute the model-loading phase of the simulator in the single-TEE setting. 

    \item \textbf{In preprocessing phase:} for every layer in the ML model,
    \begin{myenumerate}
        \item Same as the Single-TEE case.

        \item $\simr$ runs the simulator of Theorem~\ref{gen-theorem} and outputs simulated BGI key $k_b$ to the machine $M_b$.
    \end{myenumerate}

    \item \textbf{In online phase:} execute the online phase of the simulator described in the single-TEE setting. 
\end{myenumerate}

\noindent We show that the real world distributions is computationally
indistinguishable to the simulated distributions via the standard hybrid argument.

\begin{myenumerate}
    \item $Hyb_0:$ This corresponds to the real world distribution where the model-owner, the datapoint-owner, secure hardware machine $T_b^{j}, T_b^{k}$, and general-purpose machine $M_{1-b}$ execute the system as mentioned in the description of the protocol.

    \item $Hyb_{1.(i, 1)}$: In this hybrid, the $\func_{gen}$ functionality for generating the BGI keys for the $i^{th}$ layer is simulated using the simulator of $\func_{gen}$ from Theorem~\ref{gen-theorem}. The simulator of $\func_{gen}$ receives as input the BGI key $k_b$ for the $i^{th}$ layer.
    
    \item $Hyb_{0.(i,2)}$: In this hybrid, the simulator for the $i^{th}$ layer $\func_{gen}$ protocol receives as input $k_b$, where $k_b$ is generated by computing the simulator of the BGI scheme. 
    \par $Hyb_{0.(i,1)} \approx_c Hyb_{0.(i,2)}$ follows from the security of the BGI scheme. $Hyb_{0.(i,2)} \approx_{c} Hyb_{0.(i+1,1)}$, for $i < L$ where $L$ is the number of the layers in the ML model, follows from the simulation security of the $\func_{gen}$ protocol.
    
    \item Hybrids $Hyb_1$ to $Hyb_8$ are the same as described in the single-TEE setting. 
    \par $Hyb_{0.(L,2)} \approx_c Hyb_{1}$ follows from the simulation security of the DH key exchange.   
\end{myenumerate}
\section{Workloads}
\label{a:adx_models}
\label{a:models}

M-FNN consists of 3 dense (that is, fully-connected) layers with ReLU activations for 
    each layer~\cite{mohassel2017secureml, kumar2020cryptflow}. 
    I-FNN is composed of two dense layers and 
uses Tanh activations~\cite{rouhani2018deepsecure}. 
M-CNN consists of one convolution layer with 5x5 filters, 
    one average pooling layer with  2x2 pool size, 
        and two dense layers that use ReLU activations~\cite{liu2017oblivious}. 
C10-CNN  uses seven convolutional layers with 3x3 and 1x1 filters, 
        two average pooling layers with 2x2 pool 
size, and one fully-connected output layer~\cite{liu2017oblivious, mishra2020delphi}. 
Finally, C100-R32 uses 32 convolutional layers with 3x3 filters, 30 ReLU activation layers, 
    several add layers (used by shortcut paths), one global average pooling
layer, and one fully-connected output layer~\cite{mishra2020delphi}.
We omit stride and padding sizes of convolution layers here and refer the reader to prior work for their details.

\section{Pricing model}
\label{a:pricingmodel}

For network transfers, we directly use the prices reported
    by the cloud providers for inbound, outbound, and local data transfers.
To get the hourly \cpu cost, we take the hourly machine cost and split it into \cpu
cost and memory cost by making a simplifying assumption that two-third of the
total machine cost is due to \cpu and one-third is due to memory.
    In reality, pricing resources is an involved task that also depends on factors such as
business demand~\cite{awssagemakerpricing, awsinstancepricing}. Therefore, our derived resource prices should only be treated as estimates.

\emparagraph{Network pricing.}
Azure and AWS both do not charge for inbound traffic or local network transfers.
For outbound traffic, both these providers charge at 
    $\$0.5$ per GB~\cite{awsnetworkpricing, azurenetworkpricing}.

\emparagraph{\cpu pricing.}
    The hourly prices of \texttt{m5.4xlarge}, \texttt{D16s-v3}, \texttt{L8s-v2},
and \texttt{DC1s-v2} when reserved for three years are 
    $\$0.337$, $\$0.406$, $\$0.264$, and $\$0.119$, respectively~\cite{awsinstancepricing, azureinstancepricing}. Combining this 
    data with the specifications of these machines (that is, the number of \cpus
and amount of RAM), and the method described
above, we get the following per hour \cpu cost: $\$0.015$ for
\texttt{m5.4xlarge}, $\$0.017$ for \texttt{D16s-v3}, $\$0.022$ for
\texttt{L8s-v2}, and $\$0.079$ for the SGX-enabled \texttt{DC1s-v2}.

\frenchspacing

\section*{Acknowledgments}
We thank
Ishtiyaque Ahmad,
Alvin Glova,
Rakshith Gopalakrishna,
Arpit Gupta,
Abhishek Jain,
Srinath Setty,
Jinjin Shao,
Tim Sherwood,
Michael Walfish,
and Rich Wolski for feedback and comments that improved this draft.

{
\footnotesize
\begin{flushleft}
\setlength{\parskip}{0pt}
\setlength{\itemsep}{0pt}
\bibliographystyle{abbrv}
\bibliography{conferences,paper}
\end{flushleft}
}
\label{p:last}
\end{document}